\documentclass[
 prxintelligence,twocolumn,
 amsmath,amssymb,
 aps,nofootinbib,longbibliography,
]{revtex4-2}

\usepackage{graphicx}
\usepackage{booktabs}
\usepackage{mathtools}
\usepackage{amsthm}
\usepackage{bbm}
\usepackage{microtype}

\usepackage{color}
\usepackage{dcolumn}
\usepackage{bm}
\usepackage{mathtools}      
\usepackage{mathrsfs}       
\usepackage{float}
\usepackage[table]{xcolor}

\usepackage[colorlinks=true,
            linkcolor=blue,
            citecolor=blue,
            urlcolor=blue]{hyperref}
\usepackage[capitalize,noabbrev]{cleveref}

\theoremstyle{plain}
\newtheorem{theorem}{Theorem}
\newtheorem{proposition}[theorem]{Proposition}
\newtheorem{lemma}[theorem]{Lemma}
\newtheorem{corollary}[theorem]{Corollary}
\theoremstyle{definition}

\theoremstyle{remark}
\newtheorem{remark}[theorem]{Remark}

\newcommand{\cD}{{\mathcal{D}}}
\newcommand{\cE}{{\mathcal{E}}}

\newcommand{\cH}{{\mathcal{H}}}
\newcommand{\cB}{{\mathcal{B}}}

\newcommand{\cN}{{\mathcal{N}}}

\newcommand{\cL}{{\mathcal{L}}}
\newcommand{\cW}{{\mathcal{W}}}

\newcommand{\cM}{{\mathcal{M}}}
\newcommand{\cX}{{\mathcal{X}}}
\newcommand{\cY}{{\mathcal{Y}}}
\newcommand{\cZ}{{\mathcal{Z}}}
\newcommand{\cS}{{\mathcal{S}}}

\newcommand{\bx}{{\boldsymbol{x}}}

\newcommand{\bz}{{\boldsymbol{z}}}

\newcommand{\bh}{{\boldsymbol{h}}}

\newcommand{\bv}{{\boldsymbol{v}}}

\newcommand{\bth}{{\boldsymbol{\theta}}}
\newcommand{\bzet}{{\boldsymbol{\zeta}}}
\newcommand{\bC}{{\boldsymbol{C}}}

\newcommand{\bW}{{\boldsymbol{W}}}

\newcommand{\bP}{{\boldsymbol{P}}}

\newcommand{\ba}{{\boldsymbol{a}}}
\newcommand{\bb}{{\boldsymbol{b}}}

\newcommand{\bw}{{\boldsymbol{w}}}

\newcommand{\bbR}{{\mathbb{R}}}
\newcommand{\bbE}{{\mathbb{E}}}

\newcommand{\btheta}{{\boldsymbol{\theta}}}

\DeclareMathOperator{\tr}{\textup{Tr}}

\providecommand{\Tr}{}
\renewcommand{\Tr}{\operatorname{Tr}}
\newcommand{\ket}[1]{|#1\rangle}
\newcommand{\bra}[1]{\langle #1|}
\newcommand{\braket}[2]{\langle #1|#2\rangle}
\newcommand{\norm}[1]{\left\lVert #1 \right\rVert}

\newcommand{\wone}{\mathcal{W}_1}

\begin{document}

\title{Latent-Conditioned Parameterized Quantum Circuits as Universal Approximators for Distributions over Quantum States}


\author{Quoc Hoan Tran}
\email{tran.quochoan@fujitsu.com}
\affiliation{Quantum Laboratory, Fujitsu Research, Fujitsu Limited, Kawasaki, Kanagawa 211-8588, Japan}

\author{Koki Chinzei}
\affiliation{Quantum Laboratory, Fujitsu Research, Fujitsu Limited, Kawasaki, Kanagawa 211-8588, Japan}

\author{Yasuhiro Endo}
\affiliation{Quantum Laboratory, Fujitsu Research, Fujitsu Limited, Kawasaki, Kanagawa 211-8588, Japan}

\author{Hirotaka Oshima}
\affiliation{Quantum Laboratory, Fujitsu Research, Fujitsu Limited, Kawasaki, Kanagawa 211-8588, Japan}

\date{\today}


\begin{abstract}
Many applications in quantum simulation, quantum chemistry, and quantum machine learning require not a single quantum state but an ensemble of states characterizing the heterogeneity of a target system. Preparing such ensembles state-by-state is prohibitive in both
variational and fault-tolerant settings, motivating a generative-modeling approach. We introduce latent-conditioned parameterized quantum circuits
(LPQCs), a hybrid quantum-classical framework in which classical neural
networks map a latent variable sampled from a prior distribution to the
parameters of a parameterized quantum circuit. We prove that LPQCs are
universal approximators for probability measures over density
operators in the $1$-Wasserstein distance, extending classical universal
approximation theorems to the quantum-distribution setting. We additionally
introduce a multimodal latent prior and a mixture-of-experts circuit
architecture, and show empirically that the latent-conditioned parameterization alleviates the barren plateau problem during
optimization, a behavior for which we provide rigorous partial guarantees. Numerical experiments validate the framework on a synthetic multi-cluster ensemble
of mixed quantum states and on a QM9-derived ensemble of 3-D molecular structures. 
In these tasks, LPQC outperforms recent quantum generative baselines and
matches the generation quality of a classical neural-network baseline,
while requiring an output dimension that grows only linearly with the
number of qubits rather than exponentially. By leveraging classical expressivity in the
latent space, LPQCs offer a tractable route to quantum generative modeling.
\end{abstract}

\maketitle

\section{Introduction}\label{sec:intro}

Generative models for quantum data play a pivotal role in advancing quantum technologies. They approximate distributions over quantum states that are intractable classically. The central object is typically not a single state but an ensemble of
density operators, and such ensembles arise in several settings. In
quantum simulation, the finite-temperature equilibrium of a many-body
system is a Gibbs ensemble, whose preparation is a long-standing
challenge for quantum algorithms~\cite{terhal:2000:gibbs,poulin:2009:gibbs}. 
In quantum chemistry, the many distinct molecular geometries in datasets
such as QM9~\cite{ramakrishnan:2014:QM9} can each be represented as a
quantum state, so that the dataset induces an ensemble of quantum states.
Such a distribution, once learned, can support virtual screening, property
prediction, and data augmentation for hybrid quantum-classical
pipelines~\cite{rathi:2023:3dqae:encoding,wu:qvae-mole:2024}.
In quantum machine learning (QML), state ensembles serve as training data for tasks such as anomaly
detection~\cite{liu:2018:qad} and for warm-starting variational algorithms~\cite{zou:2025:npj:flow:vqe}.

Preparing such ensembles element by element is prohibitive in both near-term and fault-tolerant settings. Variational methods such as the Variational Quantum Eigensolver (VQE)~\cite{peruzzo:2014:VQE} require a fresh optimization loop with enormous circuit evaluation times for each new target. Fault-tolerant routines based on quantum phase estimation (QPE)~\cite{abrams:1999:qpe}, or block-encoded state preparation~\cite{gilyen:2019:qsvt} instead need deep circuits and substantial ancilla overhead per state, placing them beyond near-term reach. The exponential dimensionality of quantum state spaces compounds this challenge, calling for universal generative frameworks that capture multimodal structure across these settings. A generative model that, once trained, produces samples from the entire ensemble in a single forward pass would dramatically reduce this cost.

Parameterized quantum circuits (PQCs) form a cornerstone of QML, serving as ans\"atze for optimization, classification, and generative modeling. Comprising tunable gates optimized classically, PQCs power hybrid variational quantum algorithms~\cite{cerezo:2021:vqas:natrev} such as the VQE and the Quantum Approximate Optimization Algorithm (QAOA)~\cite{farhi:2014:qaoa}. Recent results establish PQCs as universal approximators of continuous functions~\cite{tran:2021:prl:uap} and multivariate distributions~\cite{barthe:2025:pqc:multivariate} via expectation-value sampling. However, applying PQCs to distributions over quantum states, such as ensembles of density matrices, remains underexplored and inefficient, hindered by barren plateaus~\cite{clean:2018:natcom:barren} where gradients vanish exponentially with system size.

We prove that latent-conditioned PQCs (LPQCs)---where classical neural networks (NNs) generate circuit parameters conditioned on latent variables---universally approximate any probability measure on the space of $n$-qubit density operators in the 1-Wasserstein distance. To enable practical training, we incorporate multimodal prior distributions in the latent space and soft attention mechanisms for multiple experts. We empirically further show that the latent prior mitigates barren plateaus by biasing initializations toward gradient-rich regions. Numerical simulations of multi-cluster mixed states and quantum-chemical datasets demonstrate the model's efficiency.

\begin{figure*}
  \centerline{\includegraphics[width=2\columnwidth]{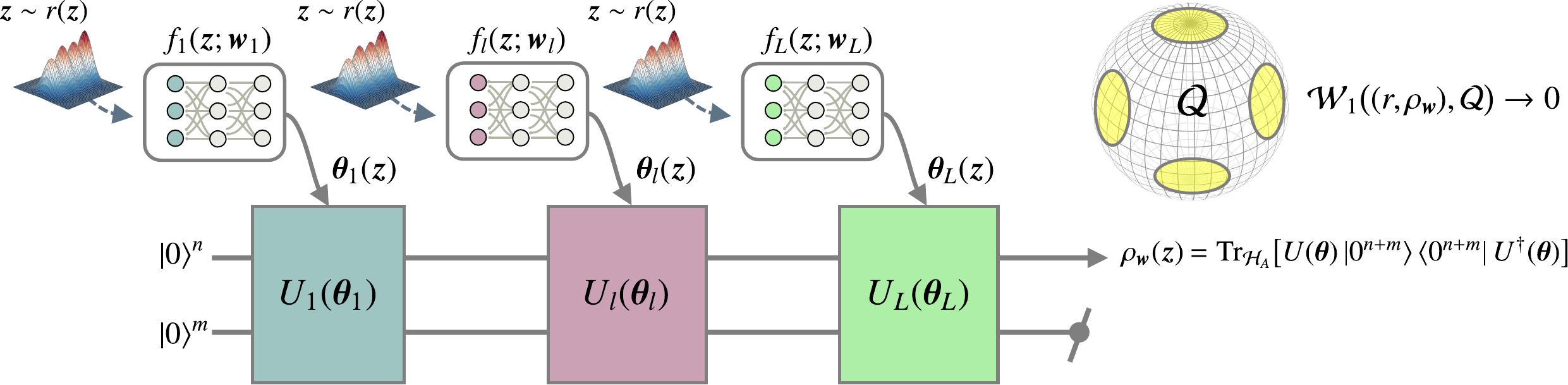}}
  \caption{Schematic of the LPQC framework for approximating the target distribution $Q$ on the density matrices space $\cD(\cH_S)$. Latent variable $\bz \sim r(\bz)$ is mapped via NNs $f_\ell(\bz; \bw_\ell)$ to gate parameters $\btheta_\ell(\bz)$, which define the unitary $U(\btheta(\bz)) = U_L(\btheta_L(\bz)) \cdots U_1(\btheta_1(\bz))$ on $\cH_S \otimes \cH_A$. In our setting, the same $\bz$ is used for all $f_\ell$. Applying $U(\btheta(\bz))$ to $\ket{\boldsymbol{0}}$ and tracing over the ancilla $\cH_A$ yields ensemble $(r, \rho_\bw)$, which can approximate any $Q$ arbitrarily well in the 1-Wasserstein distance $\cW_1$.}
  \label{fig:overview}
\end{figure*}

\section{Preliminaries}\label{sec:prelim}

Let $\cH_S\cong\mathbb{C}^{2^n}$ be the Hilbert space of an $n$-qubit system. Given the set $\cB(\cH_S)$ of bounded linear operators on $\cH_S$, we define the convex set of density operators as
$
    \cD\!(\cH_S):=\bigl\{\rho\in\cB(\cH_S)\mid \rho\ge 0,\;\Tr(\rho)=1\bigr\}.
$
For two ensembles $P,Q$ on $\cD\!(\cH_S)$ we employ the 1-Wasserstein distance:
\begin{equation}
  \wone(P,Q):=\inf_{\pi\in\Pi(P,Q)}\,\mathbb{E}_{(\rho,\sigma)\sim\pi}\Bigl[\tfrac12\,\norm{\rho-\sigma}_1\Bigr],
  \label{eq:W1}
\end{equation}
where $\Pi(P,Q)$ is the set of couplings having $P,Q$ as marginals.

\emph{The learning task.---}
Let $\cE = \{(p_j, \rho_j)\}_j$ denote an ensemble of $n$-qubit quantum states, where $\rho_j \in \cD(\cH_S)$ and $\{p_j\}$ are non-negative weights satisfying $\sum_j p_j = 1$. The ensemble induces a probability measure $Q$ on $\cD(\cH_S)$, namely $Q = \sum_j p_j \,\delta_{\rho_j}$ in the atomic case, or more generally any probability measure on $\cD(\cH_S)$ when the ensemble is continuous. Note that $Q$ is not the same object as the averaged density matrix $\overline{\rho} = \sum_j p_j \rho_j \in \cD(\cH_S)$ as averaging discards the spectral and multimodal structure.
Two ensembles with very different pure-state structures can share the same averaged density matrix.
Our goal is to construct from a finite sample of $\cE$ a generative model that approximates the underlying $Q$ in the
$1$-Wasserstein distance $\wone$ on $\cD(\cH_S)$.

\emph{Latent-conditioned PQCs (LPQC).---}
We fix a latent space $\cZ\subset\mathbb{R}^{d}$ endowed with a prior density $r(\bz)$ that is absolutely continuous with respect to the Lebesgue measure. To handle both pure states and mixed states, we introduce an ancillary Hilbert space $\cH_A \cong \mathbb{C}^{2^m}$ with $m$ qubits. An LPQC is a family of unitaries acting on $\cH_S\otimes\cH_A$ as
\begin{equation}
    U(\btheta (\bz))=U_L\bigl(\theta_L(\bz)\bigr)\cdots  U_1\bigl(\theta_1(\bz)\bigr), \quad \bz\in\cZ,
\end{equation}
where the primitive gates $U_\ell(\cdot)$ come from a fixed universal gate set.
Each angle vector $\theta_\ell(\bz) \in \bbR^{K_\ell}$ is produced by a classical neural network (NN) $f_\ell(\,\cdot\,;\bw_\ell):\cZ\to\bbR^{K_\ell}$ with parameters $\bw_\ell$ (Fig.~\ref{fig:overview}).
Given $\bw = (\bw_1,\ldots,\bw_L)$, sampling $\bz\sim r$ and applying $U(\theta (\bz))$ to $\ket{0^{n+m}}$ followed by partial trace over $\cH_A$ yields the ensemble
\begin{align}\label{eq:ensemble}
    (r,\rho_\bw) &:=\Bigl\{\,r(\bz),\;\rho_\bw(\bz):= \nonumber\\
    &\Tr_{\cH_A}\bigl[U(\theta (\bz))\ket{0^{n+m}}\bra{0^{n+m}}U^{\dagger}(\theta (\bz))\bigr]\,\Bigr\}_{\bz\in\cZ}.
\end{align}
We want to show that $(r,\rho_\bw)$ can approximate any target distribution $Q$ on $\cD\!(\cH_S)$ arbitrarily well in $\wone$.

\section{Universal Approximation Theorem (UAT) for Quantum States Distribution}\label{sec:uat}

We state and prove the main theoretical result.

\begin{theorem}[Approximate universality of LPQCs]\label{thm:main}
Let $Q$ be a probability measure on $\cD\!(\cH_S)$. Consider a compact latent space $\cZ \subset \mathbb{R}^{d}$, and a prior density $r$ that is positive, continuous, and absolutely continuous with respect to Lebesgue measure on $\cZ$. Then, for every $\varepsilon > 0$ there exist an LPQC architecture (a finite circuit depth and gate layout) and NN weights $\bw$ such that the ensemble $(r,\rho_\bw)$ defined in Eq.~\eqref{eq:ensemble} satisfies
\begin{equation}
    \wone\bigl(Q,\,(r,\rho_\bw)\bigr)\le \varepsilon \quad \textup{for some } \bw.
\end{equation}
\end{theorem}

\begin{proof}[Proof sketch]
    The proof proceeds in four steps that progressively bridge the gap between the arbitrary target distribution $Q$ on $\cD(\cH_S)$ and the ensemble $(r,\rho_\bw)$ produced by LPQC.
\begin{enumerate}
    \item \emph{Finite catalog approximation} (Sec.~\ref{sec:catalog}): we exploit the compactness of $\cD(\cH_S)$ in the trace-norm topology to replace $Q$ with a finite atomic ensemble $P^\ast=\{(q_i,\rho_i)\}_{i=1}^N$ that is $\varepsilon/3$-close to $Q$ in $\wone$.
    \item \emph{Tiling the latent space}: because the prior $r$ is nonatomic, we partition $\cZ$ into disjoint measurable slabs $\{R_i\}$ with topologically simple (hyperplane) boundaries whose probability masses match the catalog weights $q_i$, and we assign to each region a branch unitary $V_i$ that purifies the catalog state $\rho_i$ on $\cH_S \otimes \cH_A$. This defines a piecewise-constant target decoder $V(\bz)$.
    \item \emph{Smoothing}: since $V(\bz)$ is discontinuous at region boundaries, we mollify the underlying Hamiltonian $H(\bz)$ with a narrow kernel to obtain a continuous approximation $V'(\bz)$ that agrees with $V(\bz)$.
    \item \emph{Uniform PQC approximation} (Sec.~\ref{sec:express}): Lemma~\ref{lem:latent-universal} combines a Trotterized, exactly compiled Pauli decomposition of $V'(\bz)$ with the classical universal approximation theorem for the NNs $f_\ell$ to produce a circuit $U(\btheta(\bz))$ that is uniformly close to $V'(\bz)$ in operator norm. Translating this operator-norm bound into a trace-distance bound on the reduced states via the Fuchs--van de Graaf relation, and assembling a coupling that pairs each $\rho_\bw(\bz)$ with the catalog state of the cell containing $\bz$, yields a Wasserstein bound of $2\varepsilon/3$ between $(r,\rho_\bw)$ and $P^\ast$. The triangle inequality then gives the desired bound $\wone(Q,(r,\rho_\bw)) \le \varepsilon$.
\end{enumerate}
\end{proof}
The detailed steps of the proof are listed below.

\subsection{Finite Catalog Approximation}\label{sec:catalog}

The first step replaces the arbitrary target $Q$ by an atomic approximation.

\begin{lemma}[Finite catalog approximation]\label{lem:finite-support}
For every probability measure $Q$ on $\cD\!(\cH_S)$ and every $\varepsilon>0$ there exists a finite ensemble $P^\ast=\{(q_i,\rho_i)\}_{i=1}^N\subset \cD\!(\cH_S)$ such that $\wone\bigl(Q,P^\ast\bigr)\le \varepsilon$.
\end{lemma}

\begin{proof}
Since $\cD(\cH_S)$ can be equipped with the trace-norm topology, it is
compact and metrizable, hence totally bounded. With the metric
$d(\rho, \sigma) = \tfrac{1}{2}\norm{\rho - \sigma}_1$, for any
$\varepsilon > 0$ there exists a finite $\varepsilon$-net
$\{\rho_i\}_{i=1}^N$ such that the closed balls
$B(\rho_i; \varepsilon) := \{\sigma : d(\sigma, \rho_i) \le \varepsilon\}$
cover the space. We take these net points as the catalog states and
define the Voronoi cells that partition $\cD(\cH_S)$:
\begin{align}
C_i := \bigl\{\sigma : d(\sigma, \rho_i) &\le d(\sigma, \rho_j)\ \forall j, \nonumber\\
&\text{ and } d(\sigma,\rho_i) < d(\sigma,\rho_k)\ \forall k<i\bigr\}.
\end{align}
The weight assigned to each catalog state
is the $Q$-mass of its cell, $q_i := Q(C_i)$. Here, indices with $q_i = 0$ are
discarded. This defines the finite ensemble
$P^* = \{(q_i, \rho_i)\}_{i=1}^N$, with $\sum_i q_i = 1$ since the cells
partition the space.

To bound $\wone(Q, P^*)$, consider the coupling $\pi$ that samples
$\sigma \sim Q$ and maps it to $\rho_i$ whenever $\sigma \in C_i$.
Its second marginal places mass $Q(C_i) = q_i$ on each $\rho_i$, so
$\pi \in \Pi(Q, P^*)$. For every pair in the support, $\sigma$ and $\rho_i$
lie in the same cell, so $d(\sigma, \rho_i) \le \varepsilon$ by
construction. Plugging this coupling into Eq.~\eqref{eq:W1} yields
$\wone(Q, P^*) \le \varepsilon$.
\end{proof}

\subsection{Classical-Quantum Expressivity}\label{sec:express}

To convert the finite catalog into the output of a single latent-conditioned circuit, we need two ingredients: the ability to choose distinct latent regions that sample each catalog element with the desired probability, and the ability to implement within each region a unitary that prepares the associated state. The following lemma settles the latter one.

\begin{lemma}[Universal approximation of latent maps]\label{lem:latent-universal}
Let $\cZ\subset\mathbb{R}^d$ be compact and let $V(\bz) = e^{iH(\bz)}$, where $H\colon\cZ\to\cB(\cH_S\otimes\cH_A)$ is a continuous Hermitian-operator-valued map with $\sup_{\bz}\norm{H(\bz)}\le \pi$. Suppose
\begin{enumerate}
  \item Each map $f_\ell(\,\cdot\,;\bw_\ell)$ is a classical NN whose activation function satisfies the hypotheses of the classical universal approximation theorem~\cite{hornik:1991:apprx} for uniform approximation of continuous functions on compact sets (e.g., any nonconstant, bounded, continuous activation such as $\tanh$).
  \item The primitive gate set contains the parameterized single-qubit rotations $R_y(\vartheta)$, $R_z(\vartheta)$ on every qubit and the $\mathrm{CNOT}$ gate on nearest-neighbor pairs (as in the HEA of Eq.~\eqref{eqn:HEA}). In particular, it generates $\mathrm{SU}(2^{n+m})$.
\end{enumerate}
Then for every $\delta>0$ there exist a circuit depth $L = L(\delta)$, a fixed gate layout, and weights $\bw:=\{\bw_\ell\}_\ell$ such that
\begin{equation}
    \sup_{\bz\in\cZ}\,\norm{U(\btheta (\bz)) - V(\bz)} \le \delta,
\end{equation}
where $\theta_\ell(\bz):=f_\ell(\bz;\bw_\ell)$ and $U(\btheta (\bz))$ is the full circuit obtained by composing the primitive gates. In particular, $\mathbb{E}_{\bz\sim r}\norm{U(\btheta(\bz)) - V(\bz)}\le\delta$ for every probability density $r$ on $\cZ$.
\end{lemma}


\begin{proof}
\emph{Step 1: Pauli expansion.} Expand $H(\bz)$ in the Pauli basis $\{P_a\}_{a=1}^{G}$ of $\cB(\cH_S\otimes\cH_A)$, $G = 4^{n+m}$ (excluding the identity, which contributes only a global phase that we absorb into one fixed $R_z$ pair):
\begin{equation*}
    H(\bz) = \sum_{a=1}^{G} c_a(\bz)\,P_a, \qquad c_a(\bz) = 2^{-(n+m)}\Tr\!\bigl[H(\bz)P_a\bigr].
\end{equation*}
Each coefficient $c_a\colon\cZ\to\mathbb{R}$ is continuous (as a bounded linear functional of the continuous map $H$).
Since $\Tr(H^2) = \sum_{a,b}c_ac_b\Tr(P_aP_b) = 2^{n+m} \sum_a c^2_a$, it gives the Parseval identity:
\begin{align}
    \sum_a c_a^2 = \dfrac{1}{\sqrt{G}}\Tr(H^2).
\end{align}
Since $\norm{H}\le \pi$, every eigenvalue $\lambda_j$ of $H$ satisfies $|\lambda_j | \leq \pi$, hence $\Tr(H^2) = \sum_j \lambda^2_j \leq \sqrt{G}\pi^2$. Therefore, $\sum_a c_a^2 \leq \pi^2$
We define $\Gamma := \sup_{\bz}\sum_a |c_a(\bz)|$,  by the Cauchy--Schwarz inequality we have a constant bound for $\Gamma$:
\begin{align}
    \Gamma \le \sqrt{G} \left( \sum_a c_a^2 \right)^{1/2} \leq 2^{n+m} \pi.
\end{align}

\emph{Step 2: Trotterized compilation with continuous angle functions.} For $k\in\mathbb{N}$ define the first-order Trotter circuit
\begin{equation}
    T_k(\bz) := \Bigl(\prod_{a=1}^{G} e^{i c_a(\bz) P_a / k}\Bigr)^{\!k}.
\end{equation}
By the standard first-order Trotter error bound~\cite{childs:2021:trotter:prx},
\begin{equation}
    \norm{e^{iH(\bz)} - T_k(\bz)} \le \frac{1}{2k}\Bigl(\sum_a |c_a(\bz)|\Bigr)^{\!2} \le \frac{\Gamma^2}{2k},
\end{equation}
so choosing $k := \lceil \Gamma^2/\delta \rceil$ gives $\sup_\bz\norm{V(\bz) - T_k(\bz)} \le \delta/2$. 

Consider a single factor $e^{i\varphi P_a}$, where $\varphi=c_a(\bz)/k$. We show it is realized exactly by the primitive gate set. Let $\mathsf{H}$ and $\mathsf{S}$ denote the Hadamard and phase gates. For a nonidentity Pauli string
$P_a=\bigotimes_{j}\sigma_j$ with $\sigma_j\in\{I,X,Y,Z\}$, we take
a single-qubit Clifford on each qubit of the support that rotates its tensor factor onto
        $Z$ ($\mathsf{H}$ sends $X\mapsto Z$; $\mathsf{H}\mathsf{S}^\dagger$ sends $Y\mapsto Z$;
        $Z$ is left fixed), followed by
a ladder of $\mathrm{CNOT}$s that fuses the resulting string $\bigotimes_{j\in\mathrm{supp}}Z_j$
        onto a single qubit, via $\mathrm{CNOT}_{ct}\,(Z_c Z_t)\,\mathrm{CNOT}_{ct}=Z_t$.
Their product is a Clifford $C_a$ with
\begin{align}
    C_a\,P_a\,C_a^\dagger = s_a\,Z_1,\qquad s_a\in\{+1,-1\}.
\end{align}
Using the conjugation rule $C^\dagger e^{i\theta Q}C=e^{i\theta\,C^\dagger Q C}$
(for any unitary $C$ and Hermitian $Q$)
with the convention $R_z(\theta):=e^{-i\theta Z/2}$, we have
\begin{align}
    e^{i\varphi P_a}
      = C_a^\dagger\, e^{i\varphi s_a Z_1}\, C_a
      = C_a^\dagger\, R_z(-2 s_a\varphi)_1\, C_a .
\end{align}

Compiling the abstract Clifford $C_a$ into the primitive gate set realizes it only up to a global
phase, $\tilde{C}_a=e^{i\alpha_a}C_a$. The gate-by-gate inverse of the compiled circuit realizes
$\tilde{C}_a^\dagger=e^{-i\alpha_a}C_a^\dagger$, so the conjugating phases commute through the middle factor as scalars and cancel.

Concatenating the $r$ Trotter layers, each built from $G$ exactly-compiled Pauli rotations,
expresses $T_r(\bz)$ entirely in the primitive gate set with depth
\begin{equation}
    L(\delta)=O\!\bigl(r\,G\,\mathrm{poly}(n+m)\bigr),
\end{equation}
where the $\mathrm{poly}(n+m)$ factor counts the Clifford gates and $\mathrm{SWAP}$ networks inside
each $C_a$. Every gate angle in $T_r(\bz)$ is either
a constant --- the fixed Clifford-compilation angles, reproduced exactly by the output biases of the networks $f_\ell$; or
a continuous function of $\bz$ --- the variable single-qubit rotations
        $R_z(g_j(\bz))_1$ with
       $g_j(\bz)=-\frac{2s_a c_{a(j)}(\bz)}{k},$
inheriting continuity from the coefficients $c_a(\bz)$ of Step~1.

This continuity is exactly the hypothesis required to invoke the classical universal approximation.

\emph{Step 3: NN approximation of the angle functions.} By the universal approximation theorem for continuous functions on compact sets~\cite{hornik:1991:apprx}, for each variable angle there exist weights $\bw_\ell$ such that
\begin{align}
    \sup_{\bz\in\cZ}\,\lvert f_\ell(\bz;\bw_\ell)-g_\ell(\bz)\rvert\le\delta/(2L).
\end{align}

We explicate why a perturbation on each gate angle translates into an operator-norm perturbation on the whole circuit.
For clarity, we write $U_\ell(\vartheta)=e^{-i\vartheta H_\ell}$ with $\norm{H_\ell}\le 1$ (for $R_y, R_z$ the generators are $Y/2$, $Z/2$, of norm $1/2$). Given two parameter vectors $\Theta=(\vartheta_1,\dots,\vartheta_L)$ and $\Theta'=(\vartheta_1',\dots,\vartheta_L')$, we define the corresponding circuits
$U(\Theta)=U_L(\vartheta_L)\cdots U_1(\vartheta_1)$ and $U(\Theta')=U_L(\vartheta_L')\cdots U_1(\vartheta_1')$.
A telescoping expansion yields
\begin{align}
    U(\Theta)-U(\Theta')=\sum_{k=1}^L U_L(\vartheta_L)\cdots U_{k+1}(\vartheta_{k+1})\nonumber\\
    [U_k(\vartheta_k)-U_k(\vartheta_k')]U_{k-1}(\vartheta_{k-1}')\cdots U_1(\vartheta_1').
\end{align}
Because every factor except the bracket is unitary (hence norm $1$), we can bound
\begin{align}
\norm{U(\Theta)-U(\Theta')}\le\sum_{k=1}^L\norm{U_k(\vartheta_k)-U_k(\vartheta_k')}\le \sum_{k=1}^L |\vartheta_k-\vartheta_k'|.
\end{align}
The second inequality follows from Lemma~\ref{lem:lipschitz-exp} applied to the primitive-gate exponentials, with the generator norm bound $\norm{H_\ell}\le 1$.

Taking $\Theta=\btheta(\bz) = (f_1(\bz; \bw_1), \ldots,f_L(\bz; \bw_L))$ and $\Theta'=g(\bz)= (g_1(\bz), \ldots,g_L(\bz))$ (with constant angles reproduced exactly by biases, contributing zero error), we have, for every $\bz\in\cZ$, $U(\Theta') = T_k(\bz)$, then
\begin{align}
    \norm{U(\btheta(\bz))-T_k(\bz)} = \norm{U(\Theta)-U(\Theta')}\nonumber\\
    \le  \sum_{\ell=1}^L  | f_\ell(\bz; \bw_\ell) - g_\ell(\bz)|  \le L\cdot\frac{\delta}{2L}=\frac{\delta}{2}.
\end{align}
Combining with Step 2 via the triangle inequality gives $\sup_\bz\norm{U(\btheta(\bz)) - V(\bz)} \le \delta/2 + \delta/2 = \delta$, and the lemma follows.
\end{proof}

A useful corollary is the Lipschitz continuity of the matrix exponential, which we record for completeness.

\begin{lemma}[Lipschitz continuity of $e^{iX}$]\label{lem:lipschitz-exp}
For Hermitian $A,B\in\cB(\cH)$,
$\norm{e^{iA}-e^{iB}}\le \norm{A-B}$.
\end{lemma}

\begin{proof}
Write
$
  e^{iA}-e^{iB}=\int_0^1\!\frac{d}{dt}\Bigl(e^{i\bigl((1-t)A+tB\bigr)}\Bigr)\,dt
  = i\int_0^1 e^{i(1-t)A}\,(A-B)\,e^{itB}\,dt.
$
Taking the operator (spectral) norm and noting $\norm{e^{iX}}=1$ for Hermitian $X$ yields the desired inequality.
\end{proof}

\subsection{Proof for the UAT}\label{UAT:proof}
We combine the results from subsections~\ref{sec:catalog} and~\ref{sec:express} to present the final proof of the UAT.

\begin{proof}[Proof of Theorem~\ref{thm:main}]
\textbf{(i) Catalog approximation and tiling the latent space.}
We invoke Lemma~\ref{lem:finite-support} to obtain a finite ensemble $P^\ast=\{(q_i,\rho_i)\}_{i=1}^N$ satisfying the following bound:
\begin{equation}\label{eqn:catalog}
\wone(Q,P^\ast)\le \varepsilon/3.
\end{equation}
Next, we construct pairwise disjoint measurable regions $R_i\subset\cZ$ whose probability masses satisfy $\int_{R_i} r(\bz)\,d\bz = q_i$. It is important for step (iii) below that these regions can be chosen with topologically simple boundaries: an arbitrary measurable partition guaranteed by nonatomicity alone~\cite{fremlin2011measure} could have boundaries of positive Lebesgue measure, which would invalidate the mollification argument. We therefore give an explicit construction. Let $\mu(A) = \int_A r(\bz)d\bz$ denote the (nonatomic) probability measure induced by $r$, write $\bz=(z_1,\ldots,z_d)$, and define the cumulative function $F(s) := \mu\bigl(\{\bz\in\cZ : z_1 \le s\}\bigr)$. Since $r$ is a continuous density, $F$ is continuous and nondecreasing with $F(-\infty)=0$ and $F(+\infty)=1$, so by the intermediate value theorem we may choose thresholds $s_0 \le s_1 \le \cdots \le s_N$ with $F(s_0)=0$, $F(s_i)-F(s_{i-1})=q_i$, and define the slabs
\begin{equation*}
    R_i := \{\bz\in\cZ : s_{i-1} < z_1 \le s_i\}, \qquad i=1,\ldots,N,
\end{equation*}
which satisfy $\mu(R_i)=q_i$ and partition $\cZ$ up to a $\mu$-null set. Each boundary $\partial R_i$ is contained in the union of two hyperplane slices $\{z_1 = s_{i-1}\}\cup\{z_1=s_i\}$, each of which is $\mu$-null. We define the bump functions $\chi_i(\bz):=\mathbbm{1}_{R_i}(\bz)$ so that $\sum_i \chi_i(\bz)=1$ almost everywhere.

\textbf{(ii) Branch unitaries.} We choose for each $i$ a purification $\ket{\psi_i}\in\cH_S\otimes\cH_A$ such that $\Tr_{\cH_A}\ket{\psi_i}\bra{\psi_i}=\rho_i$. Consider a unitary $V_i\in U(2^{n+m})$ such that $V_i\ket{0^{n+m}}=\ket{\psi_i}$. Let $H_i$ be any Hermitian logarithm of $V_i$, i.e.~$V_i=e^{iH_i}$ with $\norm{H_i}\le \pi$. We define the target decoder map
\begin{equation}
  V(\bz):=e^{iH(\bz)}, \qquad H(\bz):=\sum_{i=1}^N \chi_i(\bz)\,H_i.
  \label{eq:decoder}
\end{equation}
Because each $\chi_i$ is bounded and $H_i$ is Hermitian, $H(\cdot)$ is measurable and $\norm{H(\bz)}\le\pi$. Moreover, for $\bz\in R_i$ we have $\chi_i(\bz)=1$ and $\chi_{j\neq i}(\bz)=0$, whence $V(\bz)=V_i$.

\textbf{(iii) Smoothing for continuity.} $V(\bz)$ is piecewise constant, hence discontinuous at region boundaries. To enable approximation, we introduce a smoothed continuous approximation $V'(\bz)$. Because of the slab construction in step (i), $H(\bz)$ depends on $\bz$ only through the first coordinate $z_1$: writing $H(\bz)=h(z_1)$ with the step function $h(s)=H_i$ for $s\in(s_{i-1},s_i]$ (extended constantly outside $[s_0,s_N]$), we mollify in one variable, $h'(s) := (h * \phi_\eta)(s)$, where $\phi_\eta\ge 0$ is a mollifier supported on $[-\eta,\eta]$ with $\int\phi_\eta = 1$. Then $H'(\bz):=h'(z_1)$ is continuous (indeed $C^\infty$ in $z_1$), satisfies $\norm{H'(\bz)}\le \pi$ (as a convex average of the $H_i$), and agrees with $H(\bz)$ whenever $\operatorname{dist}(z_1,\{s_0,\ldots,s_N\}) > \eta$. The boundary layer $B_\eta := \{\bz\in\cZ : |z_1 - s_i|\le\eta \text{ for some } i\}$ is a decreasing family of sets with $\bigcap_{\eta>0} B_\eta = \bigcup_i\{z_1=s_i\}\cap\cZ$, which is $\mu$-null. Hence $\mu_\eta := \mu(B_\eta)\to 0$ as $\eta\to 0$, and we may choose $\eta$ small enough that $\mu_\eta \le \varepsilon^2/18$. Setting $V'(\bz) := \exp(i H'(\bz))$ and using $\norm{V'(\bz)-V(\bz)}\le 2$ (both are unitary) on $B_\eta$ and $V'=V$ off $B_\eta$,
\begin{align}
    \mathbb{E}_{\bz\sim r} [\norm{V'(\bz) - V(\bz)}] \le 0 \cdot (1 - \mu_\eta) + 2 \cdot \mu_\eta \le \varepsilon^2/9.
\end{align}

\textbf{(iv) Uniform PQC approximation.} The map $V'(\bz)=e^{iH'(\bz)}$ with $H'$ continuous and $\sup_\bz\norm{H'(\bz)}\le\pi$ satisfies the hypotheses of Lemma~\ref{lem:latent-universal}. We apply the lemma with $\delta=\varepsilon^2/9$ to obtain a circuit architecture and parameters $\btheta(\bz)$ satisfying
\begin{align}
    \mathbb{E}_{\bz\sim r}\norm{U\bigl(\btheta(\bz)\bigr)-V'(\bz)}\le \varepsilon^2/9.
\end{align} 
We define $\rho_V(\bz):=\Tr_{\cH_A}[\ket{\psi(\bz)}\!\bra{\psi(\bz)}]$ and note that $\rho_\bw(\bz) = \Tr_{\cH_A}[\ket{\phi(\bz)}\!\bra{\phi(\bz)}]$, where $\ket{\phi(\bz)} = U(\btheta (\bz))\ket{0^{n+m}}$ and $\ket{\psi(\bz)} = V(\bz)\ket{0^{n+m}}$. We have
\begin{align}\label{eq:pointwise}
  &\mathbb{E} \Bigl[\tfrac{1}{2}\norm{\rho_\bw(\bz)-\rho_V(\bz)}_1\Bigr]  \nonumber\\
  &\le \mathbb{E} \Bigl[\tfrac{1}{2}\norm{\ket{\phi(\bz)}\!\bra{\phi(\bz)}-\ket{\psi(\bz)}\!\bra{\psi(\bz)}}_1\Bigr] \nonumber\\
  &\le \sqrt{2 \mathbb{E}\norm{U(\btheta(\bz))-V(\bz)}} \nonumber\\
  &\le \sqrt{2 \bigl(\mathbb{E}\norm{U(\btheta(\bz))-V^\prime(\bz)}+\mathbb{E}\norm{V^\prime(\bz)-V(\bz)}\bigr)} \nonumber\\
  &\le \sqrt{2(\varepsilon^2/9 + \varepsilon^2/9)} = 2\varepsilon/3.
\end{align}
The first inequality is from the contractivity of the trace norm under partial trace.
The second inequality is a direct application of the Fuchs--van de Graaf relation
$\tfrac12\norm{\ket\phi\!\bra\phi-\ket\psi\!\bra\psi}_1=\sqrt{1-|\braket{\phi}{\psi}|^2}$ and the bound $|\braket{\phi}{\psi}|\ge 1-\norm{U-V}$ when $\ket\psi=V\ket0$, $\ket\phi=U\ket0$, yielding $\tfrac12\norm{\ket\phi\!\bra\phi-\ket\psi\!\bra\psi}_1\le \sqrt{2\norm{U-V}}$.
For a nonnegative random variable $X$, Jensen's inequality gives $\bbE[\sqrt{X}] \leq \sqrt{\bbE[X]}$.
Therefore, $\mathbb{E} \Bigl[ \tfrac12\norm{\ket\phi\!\bra\phi-\ket\psi\!\bra\psi}_1 \Bigr] \le \mathbb{E} \Bigl[ \sqrt{2\norm{U-V}} \Bigr] \le \sqrt{2 \mathbb{E}\norm{U-V} }$.

\textbf{(v) Coupling and Wasserstein bound.} We construct a coupling $\pi$ by sampling $\bz\sim r$ and pairing $\sigma(\bz)=\rho_\bw(\bz)$ with the fixed catalog state $\tau_i=\rho_V(\bz)$ determined by the cell $R_i$ containing $\bz$. The measure form mixes differential $d\bz$ with discrete $di$ as $\pi(d\bz,di)=r(\bz)\,\mathbbm{1}_{R_i}(\bz)\,d\bz$.
Using Eq.~\eqref{eq:pointwise}, we evaluate the expected cost as
\begin{align*}
    \mathbb{E}_{\pi}\Bigl[\tfrac12\norm{\sigma(\bz)-\tau_i}_1\Bigr] = \mathbb{E}_{\pi}\Bigl[\tfrac12\norm{\sigma(\bz)-\rho_V(\bz)}_1\Bigr] \le \tfrac{2\varepsilon}{3}.
\end{align*}

Because $\wone$ is the infimum over all couplings, we have
\begin{equation}\label{eq:wone-mid}
  \wone\bigl((r,\rho_\bw),P^*\bigr) \le 2\varepsilon/3.
\end{equation}
Combining Eqs.~\eqref{eqn:catalog} and~\eqref{eq:wone-mid} via the triangle inequality gives
\begin{align*}
    \wone\bigl(Q,\,(r,\rho_\bw)\bigr)\le \wone(Q,P^\ast)+\wone(P^\ast,(r,\rho_\bw))\le \varepsilon.
\end{align*}
This completes the proof.
\end{proof}

\subsection{Smoothing via Soft Gating}\label{sec:soft}
The theoretical construction in Theorem~\ref{thm:main} relies on a target decoder map $V(\bz)$ to prepare states from the finite catalog. However, the hard transitions at region boundaries render $V(\bz)$ discontinuous, which is incompatible with the continuity hypothesis of Lemma~\ref{lem:latent-universal} invoked in step (iv) of the proof, and would prevent gradient flow through $\bz$ in any practical realization. To address this within the proof, we replace the hard indicator functions $\chi_i(\bz)$ with smooth gating weights
\begin{equation}\label{eq:soft-gating-proof}
    \widetilde\pi_i(\bz) \;=\; \operatorname{softmax}\bigl(f_{\mathrm{attn}}(\bz)\bigr)_i,
\end{equation}
where $f_{\mathrm{attn}}$ is a classical NN. 
We then replace the hard $V(\bz) = \exp\left(i \sum_i \chi_i(\bz) H_i\right)$ with a soft-attention decoder
\begin{align}\label{eq:Vsoft-defn}
    V_{\mathrm{soft}}(\bz) =  \exp\bigl(i\,H_{\mathrm{soft}}(\bz)\bigr),\quad H_{\mathrm{soft}}(\bz) =  \sum_i \widetilde\pi_i(\bz) H_i.
\end{align}

This smoothing connects hard and soft branching. Note that the indicators $\chi_i(\bz)$ are discontinuous, so they cannot be approximated uniformly by the continuous functions $\widetilde\pi_i(\bz)$. 
The correct statement is convergence in $L^1(r)$, which is all the Wasserstein argument requires. Concretely, with the slab regions $R_i$ in step~\ref{UAT:proof}(i), for any $\tau > 0$ one can choose logits $f_{\mathrm{attn}}$ (e.g., approximating steep piecewise-linear functions of $z_1$, which lie within the scope of the classical universal approximation theorem~\cite{hornik:1991:apprx}) such that $\widetilde\pi_j(\bz) \ge 1-\tau$ for all $\bz \in R_j$ at distance greater than $\tau$ from the slab boundaries. 
Since the boundary layers have vanishing $\mu$-mass as $\tau\to 0$, $\mathbb{E}_{\bz\sim r}[\max_i|\widetilde\pi_i(\bz) - \chi_i(\bz)|]$ can be made arbitrarily small along a sequence of attention networks. Moreover, $\widetilde\pi_i$ is $C^\infty$-smooth, allowing gradients $\nabla_\bz V_{\mathrm{soft}}(\bz)$ to flow throughout the proof construction. The resulting operator-norm error satisfies
\begin{align}
    &\|V_{\mathrm{soft}}(\bz) - V(\bz)\| \le \|H_{\mathrm{soft}}(\bz) - H(\bz)\| \nonumber \\
    &\le \sum_i | \widetilde\pi_i(\bz) - \chi_i(\bz) | \,\|H_i \| \le \pi\sum_i | \widetilde\pi_i(\bz) - \chi_i(\bz) |,
\end{align}
using $\|H_i\| \le \pi$. Given $\bz \in R_j$, then $\chi_j(\bz) = 1$ and $\chi_{i\neq j}(\bz)=0$, hence
\begin{align}
    \sum_i &| \widetilde\pi_i(\bz) - \chi_i(\bz) | \nonumber\\
    &= (1 - \widetilde\pi_j(\bz)) + \sum_{i\neq j}\widetilde\pi_i(\bz)\\
    &= 1 + \sum_i \widetilde\pi_i(\bz) - 2\widetilde\pi_j(\bz) \\
    &= 2(1 - \widetilde\pi_j(\bz)) \le 2\max_i |\widetilde\pi_i(\bz) - \chi_i(\bz)|.
\end{align}
Therefore,
\begin{align}
\mathbb{E}_{\bz\sim r} \|V_{\mathrm{soft}}(\bz) &- V(\bz)\| \nonumber\\
&\le 2\pi \,\mathbb{E}_{\bz\sim r} [ \max_i |\widetilde\pi_i(\bz) - \chi_i(\bz)| ] \longrightarrow 0.
\end{align}
Therefore, $V_{\mathrm{soft}}$ inherits the universality of $V$ within the proof of Theorem~\ref{thm:main} while remaining fully differentiable.

\begin{figure*}[t]
  \centering
  \includegraphics[width=2\columnwidth]{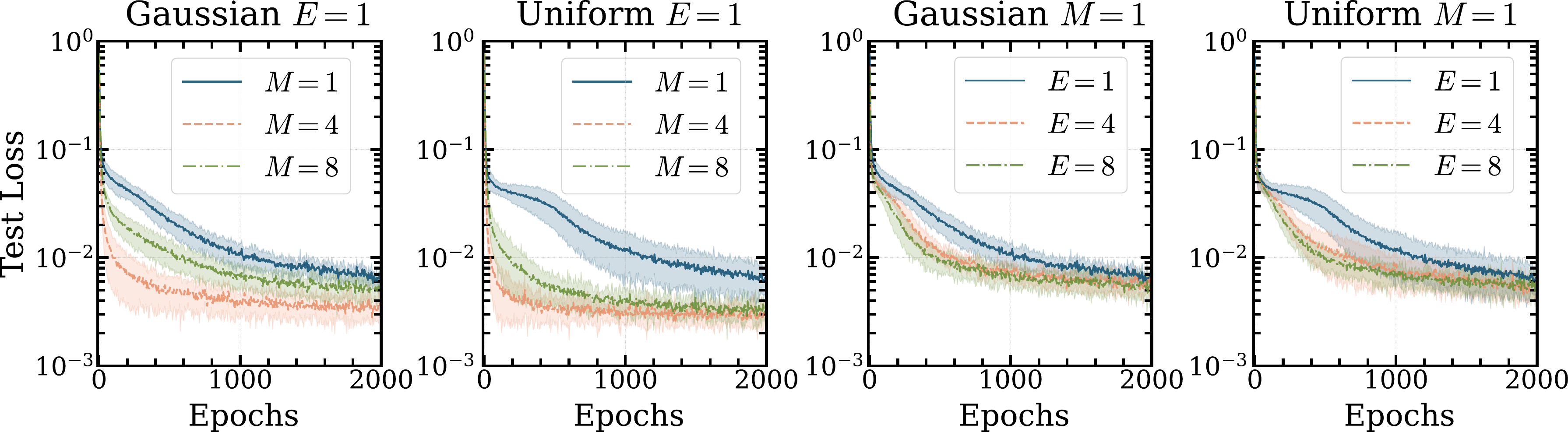}
  \caption{The performance of LPQC in learning the distribution of multi-clustered density matrices ($n=4$ data qubits) across various configurations ($M$ prior modes and $E$ experts)  with $L=10$ PQC layers. Results are averaged across 10 experimental trials per configuration, with solid and dashed lines representing the mean and shaded regions indicating the standard deviation.}
  \label{fig:compare:mod-expert}
\end{figure*}

\section{Practical Implementation}\label{sec:impl}

\subsection{Mixture of Experts}\label{sec:moe}

The mixture-of-experts (MoE) paradigm, originally introduced in classical machine learning~\cite{jacobs:1991:moe,jordan:1994:hierarchical:moe}, has recently regained prominence as a scaling strategy for large NNs~\cite{shazeer:2017:sparse:moe,fedus:2022:switch}. 
Instead of a single monolithic network, the model maintains a collection of $E$ specialized subnetworks (\emph{experts}), and a learned gating function decides---for each input---which experts to consult and with what weights. 
The motivation is twofold. First, when the target function or distribution has heterogeneous structures (e.g., distinct modes, clusters, or regimes), forcing a single network to fit them all leads to interference between regimes and slower training. 
Routing each input to a specialized subset of experts mitigates this. Second, MoE provides a form of conditional computation: model capacity can be added by increasing $E$ without proportionally increasing the cost per input, since only a few experts contribute meaningfully to any given output.

In our setting, both motivations apply. The target $Q$ is a distribution over density matrices that is, in many cases of interest (e.g., the multi-cluster ensembles of Sec.~\ref{sec:demo} or chemically diverse molecular datasets), genuinely multimodal. In this case, different regions of the latent space $\cZ$ should map to qualitatively different families of quantum states. A single MLP $f(\bz;\bw)$ producing the full parameter vector $\btheta(\bz)$ must then encode all of this heterogeneity in one set of weights, which empirically slows convergence and limits final fidelity. The MoE architecture instead allocates $E$ expert NNs $\{f^{(i)}(\bz;\bw^{(i)})\}_{i=1}^E$, each free to specialize on a portion of $\cZ$, and uses a soft gating signal $\pi_i(\bz)$ to decide which experts contribute to producing the circuit parameters at any given $\bz$. This is the structural counterpart to using a multimodal prior $r(\bz)$. The prior partitions the latent space into mode-aligned regions on the input side of the parameter-generation map, while the MoE partitions capacity on the output side.

The soft gating mechanism $\widetilde\pi_i(\bz)$ introduced in Sec.~\ref{sec:soft} has a natural counterpart in the practical training architecture. 
We reuse the same form of softmax gating, but apply it to a different object. Whereas the soft-attention decoder $V_{\mathrm{soft}}$ uses $\widetilde\pi_i(\bz)$ to mix branch Hamiltonians inside a single exponential [Eq.~\eqref{eq:Vsoft-defn}], the MoE architecture uses gating weights to mix circuit parameters produced by separate expert networks. To keep this distinction explicit, we denote the MoE gating by
\begin{equation}\label{eq:soft-gating-moe}
    \pi_i(\bz) = \frac{\exp \eta_i(\bz)}{\sum_{j=1}^{E}\exp \eta_j(\bz)}, \quad \eta(\bz) = f_{\mathrm{attn}}(\bz;\bw_{\mathrm{attn}}) \in \mathbb{R}^E.
\end{equation}

Given $E$ expert NNs producing parameter vectors $\btheta^{(i)}(\bz) \in \mathbb{R}^{\sum_\ell K_\ell}$, the MoE circuit is
\begin{equation}\label{eq:Usoft}
    U_{\mathrm{soft}}(\btheta(\bz)) \;=\; U\!\left(\,\sum_{i=1}^E \pi_i(\bz)\,\btheta^{(i)}(\bz)\,\right),
\end{equation}
i.e., the experts produce parameter vectors, the gating weights $\pi_i(\bz)$ mix them, and the resulting vector $\btheta(\bz) = \sum_i \pi_i(\bz)\btheta^{(i)}(\bz)$ is fed into a single-layered ansatz $U(\cdot)$. The experts can be trained jointly or pre-trained independently and softly combined when subsets of the data are available. The map $\bz \mapsto \btheta(\bz) \mapsto U_{\mathrm{soft}}(\btheta(\bz))$ is continuous and differentiable, so gradients flow through both the experts and the gating network.

We stress that Eq.~\eqref{eq:Usoft} is not operator-equivalent to the soft-attention decoder in Eq.~\eqref{eq:Vsoft-defn}.
The two coincide only in degenerate cases (e.g., commuting generators or a single layer). The MoE architecture is therefore motivated on architectural grounds---smoothness, specialization across modes of $Q$, and gradient flow through $\bz$---rather than as an operator-level realization of $V_{\mathrm{soft}}$.

Even though $U_{\mathrm{soft}} \ne V_{\mathrm{soft}}$, the MoE construction still inherits universality from Lemma~\ref{lem:latent-universal}, by a simple containment argument. Setting the gating logits to a constant ($f_{\mathrm{attn}} \equiv 0$, so $\pi_i = 1/E$) and all experts equal to a common network $f$, the composite map
\begin{equation}\label{eq:moe-composite}
    \bz \;\longmapsto\; \btheta(\bz) \;=\; \sum_{i=1}^E \pi_i(\bz)\,\btheta^{(i)}(\bz)
\end{equation}
reduces to $\btheta(\bz) = f(\bz)$, i.e., the MoE family contains the single-expert family for which Lemma~\ref{lem:latent-universal} was established. The universal approximation guarantee therefore transfers to $U_{\mathrm{soft}}$ without modification. More generally, Eq.~\eqref{eq:moe-composite} is a composition of the expert NNs $f^{(i)}(\bz;\bw^{(i)})$, the gating network $f_{\mathrm{attn}}(\bz;\bw_{\mathrm{attn}})$, and smooth operations (softmax, multiplication, summation), so the full MoE class consists of continuous, almost-everywhere-differentiable maps $\bz \mapsto \btheta(\bz)$, and gradients flow through both the experts and the gating network. 

\subsection{Multimodal Mechanisms}

The MoE structure is sufficient if $Q$ is smooth, the latent map $\bz \mapsto \btheta(\bz)$ does not need discrete branches, and the NNs $f_\ell$ can directly learn smooth variations. However, when $Q$ is multimodal, the basic ansatz could approximate this distribution via NNs learning jumps in $\btheta(\bz)$, but discontinuities hinder gradient flow, leading to poor optimization.
Specifically, we can consider $r(\bz)$ as a multimodal prior (mixture with $M$ modes) $r(\bz) = \sum_{i=1}^Mc_i p_i(\bz)$, where $\{c_i\}$ are non-negative mixing weights ($\sum c_i = 1$), and each $p_i(\bz)$ is a component distribution.
For example, with Gaussian, Uniform, and Beta mixtures, we can assign $p_i(\bz) = \mathcal{N}(\bz \mid \boldsymbol{\mu}_i, \boldsymbol{\Sigma}_i)$, $p_i(\bz) = \mathcal{U}(\bz\mid \ba_i, \bb_i)$, and $p_i(\bz) = \mathcal{B}(\bz\mid \boldsymbol{\alpha}_i, \boldsymbol{\beta}_i)$, respectively.
While we assign uniform $c_i$ and fix each component $p_i(\bz)$ in our experiments, these factors can be trainable in location, scale, and shape, enabling the prior to adapt during optimization.

In our implementation, we use a basic ansatz for PQC $U(\btheta) = U_L(\theta_L) \cdots U_2(\theta_2) U_1(\theta_1)U_0(\theta_0)$ as a hardware-efficient ansatz (HEA) for each layer:
\begin{align}
U_0(\btheta_0) &= \prod_{p=1}^{n+m} R_p(\theta_{0,p}),\label{eqn:HEA}\\
U_\ell(\btheta_\ell) &= \Biggl( \prod_{p=1}^{n+m} R_p(\theta_{\ell,p}) \Biggr) \Biggl( \prod_{p=1}^{n+m-1} \mathrm{CNOT}_{p,p+1} \Biggr)\nonumber,
\end{align}
where $R_p(\theta_{\ell,p}) = \exp(-i Y\theta^{(1)}_{\ell,p}/2)\exp(-i Z\theta^{(2)}_{\ell,p}/2)$ on physical qubit $p$ and
$\mathrm{CNOT}_{p,p+1}$ forms a linear nearest-neighbor coupling.
Here, $Y$ and $Z$ are single-qubit Pauli $Y$ and Pauli $Z$ matrices, respectively.
The number of parameters per layer is $K_\ell=2(n+m)$ and those parameters are produced by NNs $\btheta_\ell(\bz) = f_\ell(\bz; \bw_\ell) \in \mathbb{R}^{K_\ell}$.

\begin{figure*}[t]
  \centering
  \includegraphics[width=2\columnwidth]{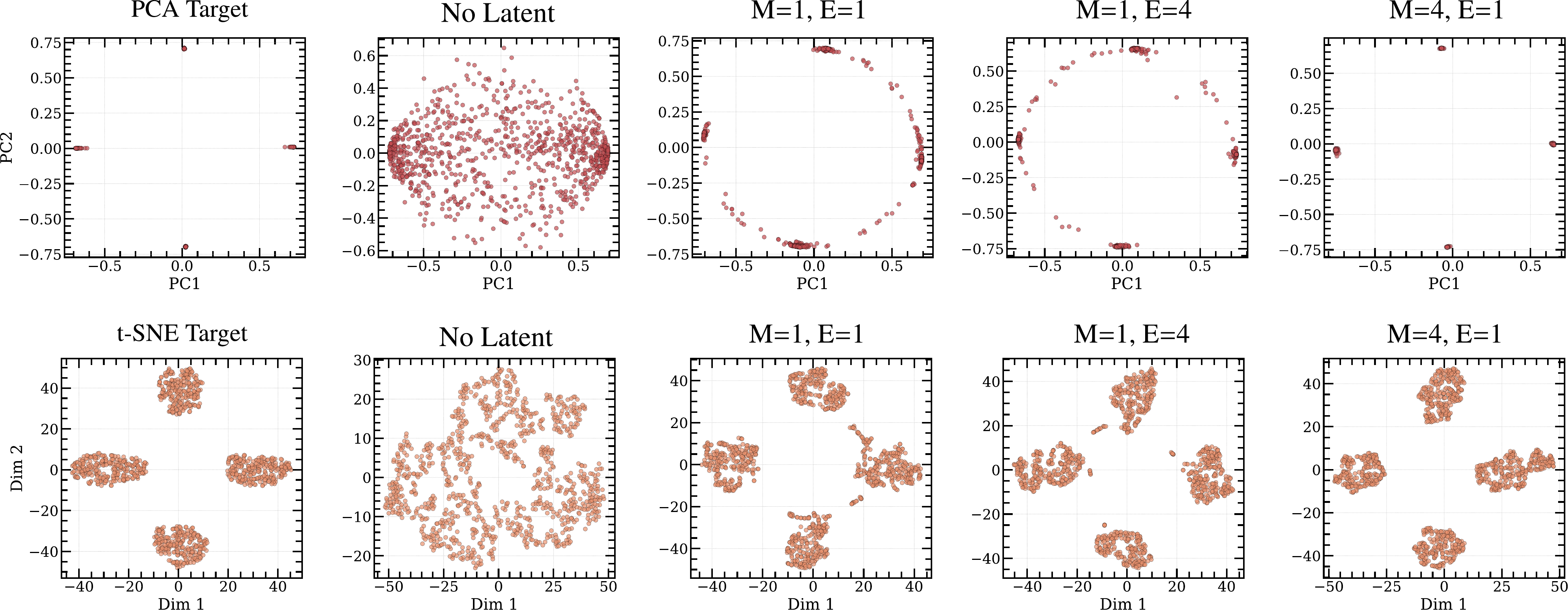}
  \caption{Visualization of target and generated ensembles using PCA and t-SNE for trained LPQC (same conditions in Fig.~\ref{fig:compare:mod-expert}) with $M$ prior modes in the Uniform distribution, and $E$ experts. We perform tomographies on each density matrix, convert them to vector form, and then apply PCA and t-SNE for dimensionality reduction.}
  \label{fig:visual:cluster}
\end{figure*}

\subsection{End-to-End Learning}
Direct minimization of $\wone(Q, (r, \rho_\bw))$ is intractable, so we use a loss $\cD_{\mathrm{Wass}}(\cX, \cY)$ approximating the 1-Wasserstein distance between density matrix ensembles $\cX=\{\rho_i\}$ and $\cY=\{\sigma_j\}$ via
a symmetric, positive definite quadratic kernel $\kappa(\rho_i, \sigma_j)$. With normalized $\kappa$ (i.e., $\kappa(\rho, \rho) = 1\, \forall\rho\in\cD(\cH_S)$), the cost matrix $\bC = (C_{i,j})\in\mathbb{R}^{|\cX|\times|\cY|}$ is $C_{i,j} = 1 - \kappa(\rho_i, \sigma_j)$. The distance solves an optimal transport linear program for plan $\bP = (P_{i,j})\in\mathbb{R}^{|\cX|\times|\cY|}$:
\begin{align}
\cD_{\mathrm{Wass}}(\cX, \cY) &= \min_{\bP} \sum_{i,j} P_{i,j}C_{i,j}, \nonumber \\
\text{s.t. }\bP \mathbf{1}_{|\cY|} &= \ba,\quad \bP^\top \mathbf{1}_{|\cX|} = \bb,\quad \bP\ge 0.
\end{align}
Here, $\mathbf{1}_{|\cX|}$ and $\mathbf{1}_{|\cY|}$ are all-ones vectors with size $|\cX|$ and $|\cY|$, respectively, and $\ba$ and $\bb$ are histograms for $\cX$ and $\cY$ (typically uniform: $\ba=\tfrac{1}{|\cX|}\mathbf{1}_{|\cX|}$, $\bb=\tfrac{1}{|\cY|}\mathbf{1}_{|\cY|}$).

Given target samples $\cY=\{\sigma_j\} \sim Q$, we sample $\bz\sim r(\bz)$ and then compute $\pi(\bz)$, $\btheta(\bz)$, and $\rho_\bw(\bz)$ to form an ensemble $\cX = \{\rho_i\}$. We use the kernel based on the super-fidelity~\cite{jaros:2009:superfid} $\kappa(\rho, \sigma) = \tr(\rho\sigma) + \sqrt{[1 - \tr(\rho^2)][1 - \tr(\sigma^2)]}$, an upper-bound approximation of fidelity due to its efficiency for mixed states, needing fewer copies.
The test loss is the Wasserstein distance $\cD_{\mathrm{Wass}}$, and the train loss combines $\cD_{\mathrm{Wass}}$ with an entropy regularization as
$\cL = \cD_{\mathrm{Wass}} + \lambda \sum_{i=1}^E \pi_i\log \pi_i$,
where $\lambda > 0$.

\section{Demonstration}\label{sec:demo}

\subsection{Learning Multimodal Distributions}
We employ a synthetic dataset of quantum density matrices designed to exhibit multimodal structure, facilitating the evaluation of LPQC as a generative model in capturing diverse quantum states.
The dataset is structured around four fixed pure states serving as cluster centers, chosen to represent a mix of product and highly entangled states: $\ket{0\ldots0}$, $\ket{1\ldots 1}$, $\ket{\mathrm{GHZ}^{+}} = \tfrac{1}{\sqrt{2}}(\ket{0\ldots0} + \ket{1\ldots 1})$, and $\ket{\mathrm{GHZ}^{-}} = \tfrac{1}{\sqrt{2}}(\ket{0\ldots0} - \ket{1\ldots 1})$, with corresponding density matrices $\rho_0$, $\rho_1$, $\rho_2$, and $\rho_3$.
Around each center, we generate a cluster of $N_s/4$  mixed states ($N_s=2048$) by applying one layer of HEA defined in Eq.~\eqref{eqn:HEA}  to the full system ($n+m$ qubits) with rotation angles randomly generated from $\cN(0, 0.05)$, then tracing out $m$ ancilla qubits (same as in LPQC).
The four centers are well separated in Hilbert space, creating distinct modes of distribution.

We explain the structure used in LPQC.
For the PQC component, we employ the HEA defined in Eq.~\eqref{eqn:HEA}, where each parameter vector
$\btheta_\ell(\bz) \in \bbR^{K_\ell}$ is generated from a multi-layer perceptron (MLP) $f_\ell(\bz; \bw_\ell)$ with tanh activation.
In the single-expert scenario, to reduce the number of training parameters, we use the shared-parameter technique, where an MLP $f(\bz; \bw)$ produces the total vector $\btheta(\bz) \in \bbR^K$ ($K=\sum K_\ell$).
This MLP has an input dimension of $d$ (corresponding to the latent dimension of $\bz$), $h$ hidden layers of dimension $D$, and an output dimension of $K$. Thus, the layer dimensions are structured as $(d, D, \ldots, D, K)$. Here, we denote the overall architecture of $f(\bz; \bw)$ compactly as $\mathrm{MLP}(d, D^{(h)}, K)$.

While the prior $r(\bz)$ can be trainable, in our implementation, we assume $r(\bz)$ is Gaussian or Uniform with fixed parameters.
For the single-mode case ($M=1$), the Gaussian is the multivariate normal distribution $\mathcal{N}(\mathbf{0}, \mathbf{I}_d)$, while the uniform is the flat distribution over the hypercube $[-1/e, 1/e]^d$.
For the multimodal case where $r(\bz)$ is a mixture with $M$ modes, we use uniform mixing weights ($c_i = 1/M$) and component distributions $\mathcal{N}(\boldsymbol{\mu}_i, e^{-2} \mathbf{I}_d)$ (for Gaussian) or uniform over the hypercube centered at $\boldsymbol{\mu}_i$ with side length $2e^{-2}$ (i.e., $[\boldsymbol{\mu}_i - e^{-2} \mathbf{1}, \boldsymbol{\mu}_i + e^{-2} \mathbf{1}]$) (for Uniform). The mean vectors $\boldsymbol{\mu}_i$ are sampled randomly from $\mathcal{N}(\mathbf{0}, \mathbf{I}_d)$.

In the MoEs scenario, instead of training $\btheta^{(i)}$ ($i=1,\ldots, E$) independently, we jointly train all parameters in all experts $\mathrm{MLP}(d, D^{(h)}, K)$ $f^{(i)}(\bz; \bw)$ and the attention network $f_{\text{attn}}(\bz; \bw_{\text{attn}})$ to produce the soft-attention weights $\eta(\bz)$ (obtained via softmax on the output logits).
Here, $f_{\text{attn}}(\bz; \bw_{\text{attn}})$ is an MLP with architecture $\mathrm{MLP}(d, 32^{(1)}, E)$ and tanh activation.
We split the train/test dataset 50-50, use the Adam optimizer with a learning rate of $0.001$ and batch size of $128$, train for $2000$ epochs using the entropy regularization with $\lambda=0.01$ (in the MoEs), and evaluate using $\cD_{\mathrm{Wass}}$.
We use TensorCircuit-NG library~\cite{zhang:2023:tensorcircuit} to simulate PQCs, and JAX~\cite{jax:2018:github} to facilitate automatic differentiation.

\begin{figure}[t]
  \centering
  \includegraphics[width=\columnwidth]{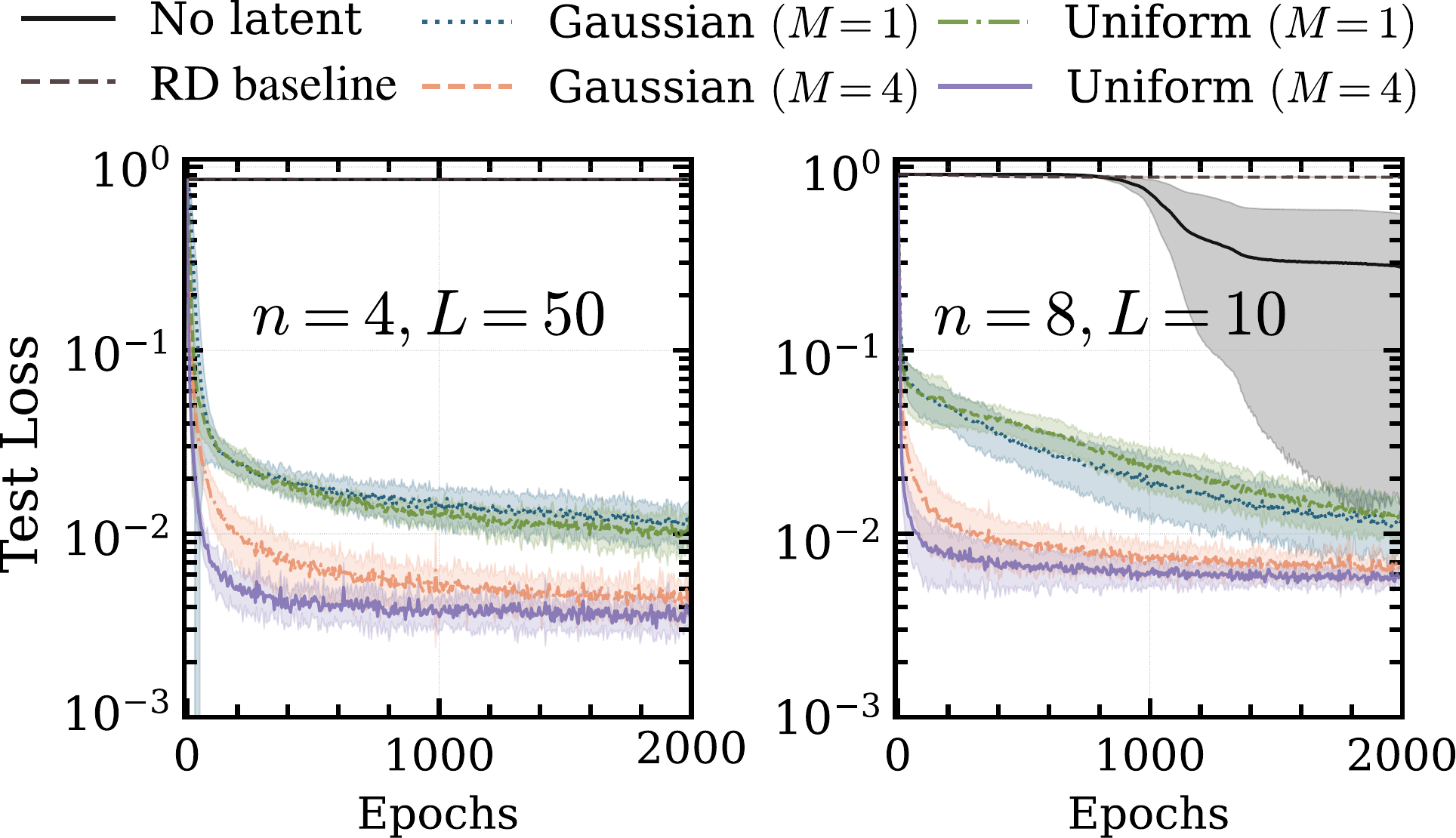}
  \caption{The test loss in training LPQC to learn the distribution of multi-clustered density matrices with $n$ qubits and $L$ PQC layers across configurations of $r(\bz)$. Results are averaged across 10 experimental trials per configuration, with solid lines representing the mean and shaded regions indicating the standard deviation.}
  \label{fig:compare:modes}
\end{figure}

\begin{figure*}[t]
  \centering
  \includegraphics[width=2\columnwidth]{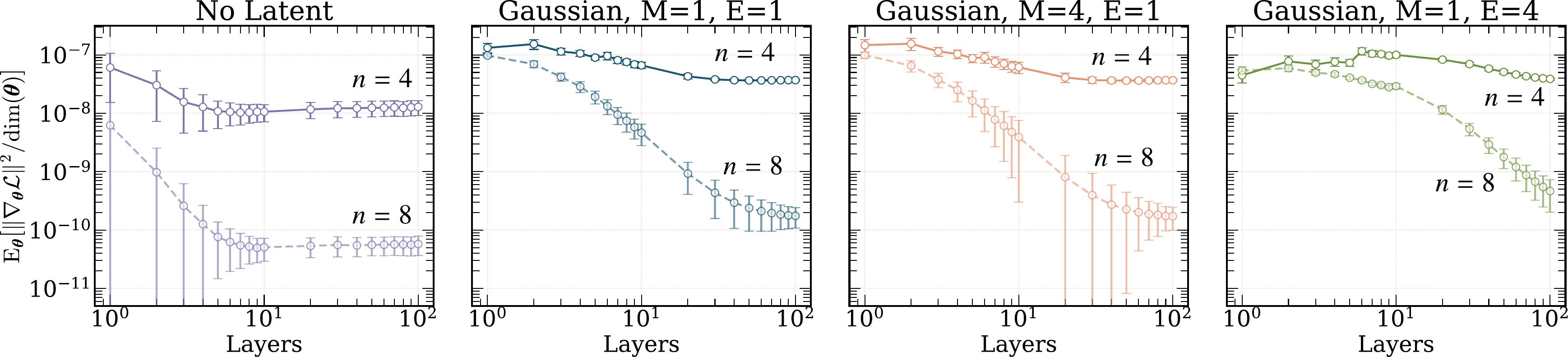}
  \caption{The average squared gradient norm, normalized by the dimension of $\btheta$, is computed by averaging over 1{,}024 independent trials sampled uniformly across the PQC parameter space. Results are presented for varying numbers of layers and qubits across different LPQC configurations (no-latent prior and Gaussian prior), with error bars indicating standard deviations.}
  \label{fig:compare:BP}
\end{figure*}

Figure~\ref{fig:compare:mod-expert} presents the test loss during training of LPQC models with $L=10$ HEA layers, $n=4$ data qubits, and $m=2$ ancilla qubits, under varying configurations of $r(\bz)$ and MoE setups.
Both multimodal priors and MoEs improve over the single-mode and single-expert cases, with multimodal priors yielding faster convergence and lower final losses.
For visual comparison of the generated and target ensembles, Fig.~\ref{fig:visual:cluster} applies dimensionality reduction via PCA and t-SNE to the tomography vectors of density matrices, revealing that $M=4$ modes in Uniform priors produce ensembles with shapes more closely matching the multimodal target (four modes).
These results underscore the value of inductive bias via multimodal priors for efficient learning of distributions over density matrices.
Jointly trained MoEs provide modest improvement, potentially attributable to the expanded NN capacity for $\btheta$ generation.
We also include a no-latent prior baseline, where $\btheta$ is directly sampled from a trainable Gaussian $p(\btheta) = \mathcal{N}(\btheta \mid \boldsymbol{\mu}, \boldsymbol{\Sigma})$.
The baseline completely fails to capture the properties of the target distribution.

In Fig.~\ref{fig:compare:modes}, we further note that the LPQC still works well even in high-depth PQC ($L=50$ layers) and a large number of qubits ($n=8$). Motivated by typical practice in quantum generative modeling~\cite{rudolph:2022:wasserstein,qute:2020:hybrid:gen}, along with the no-latent prior baseline, we include a hybrid random deterministic (RD) baseline for comparison. The first $L/2$ HEA layers use random parameters sampled from a Gaussian mixture prior with $M$ modes, while the remaining $L/2$ layers use trainable parameters. Both baselines fail to achieve effective optimization, with final test losses more than two orders of magnitude larger than those of LPQC under the same experimental setting. Traditional hybrid RD structures are primarily designed to learn distributions over classical data. In contrast, the LPQC approach provides a stronger inductive bias specifically suited to this quantum-to-quantum generative benchmark.

\emph{Barren Plateau Analysis. }
The observation in Fig.~\ref{fig:compare:modes} ties into the barren plateau phenomenon~\cite{clean:2018:natcom:barren}, where flat loss landscapes in PQCs lead to gradients that vanish exponentially, necessitating high measurement resources for accurate gradient estimation. Introducing a latent space to parameterize $\btheta$ empirically alleviates this issue, as shown in Fig.~\ref{fig:compare:BP}. Here, we compute the average squared gradient norm via the expectation $\mathbb{E}_{\btheta}\left[ \| \nabla_{\btheta}\cL \|^2 / \dim(\btheta) \right]$ over 1024 independent trials in the PQC parameter space.
In the no-latent prior case, parameters are sampled uniformly, yielding Haar-random unitaries as the number of layers $L$ increases, which causes $\mathbb{E}_{\btheta}\left[ \| \nabla_{\btheta}\cL \|^2 / \dim(\btheta) \right]$ to decay rapidly with $n, L$ and saturate at extremely low values (e.g., around $10^{-10}$ by $n=8, L=10$). By contrast, the latent space distorts the $\btheta$ distribution away from uniformity, resulting in significantly higher gradient norms across all $L$. For instance, at $n=8, L=10$, LPQCs improve the gradient norm by up to two orders of magnitude over the no-latent baseline. Here, we display only the results from the Gaussian $r(\bz)$, since the Uniform shows the same trend.

\begin{figure}[t]
  \centering
  \includegraphics[width=0.9\columnwidth]{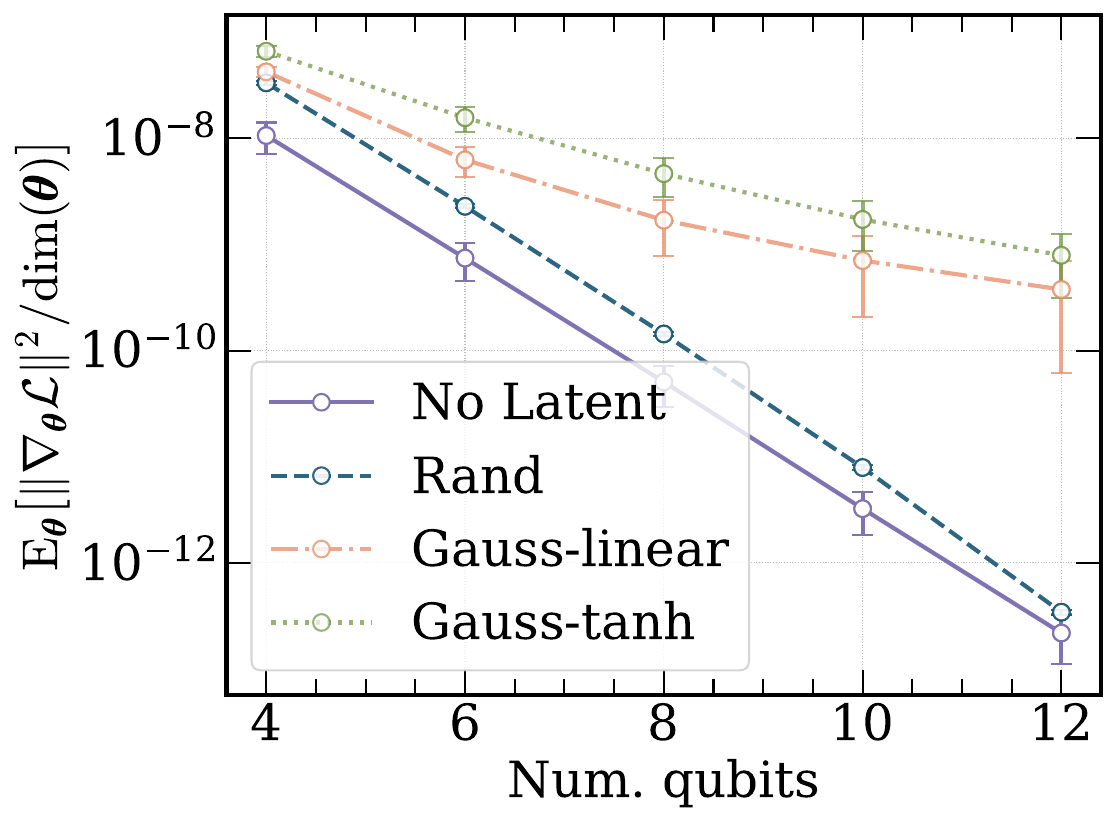}
  \caption{The average squared gradient norm as a function of the number of qubits $n$ at fixed depth $L = 10$, averaged over 1024 independent trials sampled uniformly across the PQC parameter space. Error bars indicate standard deviations. 
  Both LPQC variants (`Gauss-linear' and `Gauss-tanh') maintain gradient norms 1--2 orders of magnitude larger than the no-latent and Rand (RD) baselines across all system sizes.
  }
  \label{fig:gradnorm:qubits}
\end{figure}

To assess how this empirical alleviation scales with system size, we extend the gradient-norm experiments to $n \in \{4, 6, 8, 10, 12\}$ at fixed depth $L = 10$, including the RD baseline for comparison. We consider two LPQC variants: Gauss-linear, whose NN $f_\ell$ uses linear activation (subsuming the affine generator $\btheta_\ell = \bW_\ell \bz + \bb_\ell$ of Ref.~\cite{barthe:2025:pqc:multivariate} as a special case), and our default Gauss-tanh. Figure~\ref{fig:gradnorm:qubits} shows two trends. Both LPQC variants maintain gradient norms one to two orders of magnitude above either baseline at every $n$, with the gap stable as $n$ grows. Meanwhile, the RD baseline decays exponentially with $n$ just like the no-latent case, offering only a small constant-factor improvement. This suggests a qualitatively different inductive bias, though the evidence is empirical and limited to the regimes tested.

Our observation is consistent with empirical results using NNs to mitigate the barren plateau~\cite{lucas:pra:2022:avoidBP}. Recent theoretical analysis~\cite{yi:2025:geom:opt:Lie} adopts a geometric perspective, suggesting that $\bth$ generated by NNs follow smooth, efficient trajectories that avoid the flat regions.
We further expect that multimodal priors or MoEs with attention can sharpen these geodesic trajectories by concentrating $\btheta$ in higher-curvature subspaces, thereby explaining the improvements in gradient norm.
In Appendix~\ref{appx:BP-discussion} we discuss the Lie-algebraic perspective on why the latent prior alleviates barren plateaus, and we prove rigorous partial guarantees. 
Building on the Gaussian-initialization bounds of Ref.~\cite{zhang:2022:gaussian:init}, Proposition~\ref{prop:concentrated} and Corollary~\ref{cor:bp-free} show that the LPQC family contains configurations, realized by polynomial-size classical generators, whose expected squared gradient norm is at most polynomially small.
Proposition~\ref{prop:meta} shows that the meta-optimization over the generator weights does not face a barren plateau at such initializations.

\begin{figure}[t]
  \centering
  \includegraphics[width=\columnwidth]{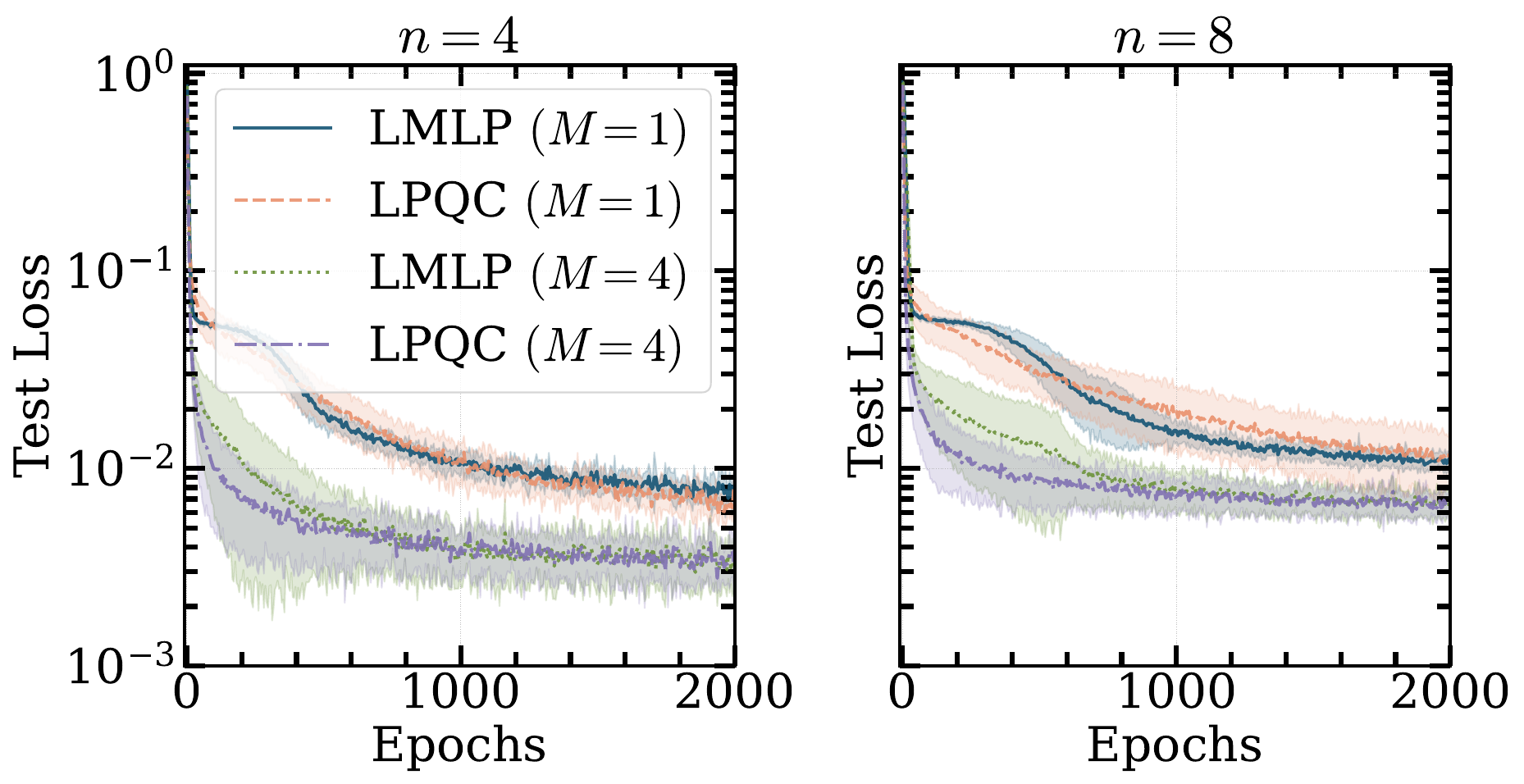}
  \caption{Test Wasserstein loss during training on the multi-cluster task, comparing the classical LMLP baseline against LPQC for two latent-prior configurations ($M=1$ unimodal and $M=4$ multimodal) at $n=4$ and $n=8$ qubits. Solid and dashed lines denote the mean over 10 trials. Shaded regions indicate the standard deviation. Both methods reach comparable final losses, while LPQC converges faster in the early training phase, especially with the multimodal prior ($M=4$).}
  \label{fig:LMLP_cluster}
\end{figure}

\emph{Comparison with a Classical Generative Baseline.---}
We investigate whether the LPQC offers any concrete advantage over a fully classical generative model on the same task. To assess this, we also implement a Latent-MLP (LMLP) baseline that mirrors the LPQC pipeline as closely as possible, excluding the quantum component. A latent vector $\bz$ is sampled from the same prior family $r(\bz)$ used by LPQC and mapped through an MLP $g(\bz;\bw_g)$ to produce a classical representation of a density matrix. We consider two variants: (i) the MLP outputs a complex matrix $A(\bz)$ that is projected to a valid density matrix via $\rho(\bz) = A(\bz)A^{\dagger}(\bz)/\Tr[A(\bz)A^{\dagger}(\bz)]$, used when LPQC has nonzero ancilla qubits ($m>0$); and (ii) when $m=0$, the MLP outputs a normalized complex state vector that defines a pure state directly. Both variants use the same training and evaluation as LPQC.

For the synthetic multi-cluster task, we match the latent dimension $d=4$, hidden dimension $D=32$, and $h=2$ hidden layers between LPQC and LMLP. Figure~\ref{fig:LMLP_cluster} compares the test Wasserstein loss during training for $n\in\{4,8\}$ and two prior configurations. Both methods reach comparable final losses, while LPQC converges faster across most epochs---especially with the multimodal prior ($M=4$). The structural advantage of LPQC lies in scaling: the LMLP generator produces an output of dimension $O = 2^{2n+1}$ to span the vectorized density-matrix space, whereas LPQC requires only $2L(n+m)$ PQC parameters. For $L=10$, $n=8$, $m=2$, this corresponds to roughly a $655\times$ reduction in output dimensionality, and the gap grows exponentially with $n$. LMLP produces classical vectorizations of density matrices rather than actual quantum states, making it unsuitable for downstream tasks that require quantum states.

\subsection{Learning Quantum Chemical Distribution}

We apply our method to the QM9 dataset~\cite{ramakrishnan:2014:QM9}, a well-established benchmark for applications in computational chemistry, including predicting molecular properties and generating 3-D structures. The dataset comprises approximately 134{,}000 compact organic molecules, each with a maximum of 9 heavy atoms (C, N, O, and F) and up to 29 atoms total, along with associated molecular properties and three-dimensional coordinates.

We align with previous work using the Incremental Many-body Projected Ensemble (IMPE) framework~\cite{tran:2025:QUAP} on a curated subset of the QM9 dataset.
Every 3-D molecule in this selection is represented as a 7-qubit quantum state (Appendix~\ref{appx:encoding:3D}), which leads to learning the distribution of pure quantum states linked to the QM9 subset.
Therefore, LPQC has $n=7$ data qubits with $m=0$ ancilla qubits.
We also use the same MLPs described in the previous section.

We first consider molecules with precisely 7 heavy atoms and 2 unique ring systems, yielding a collection of 488 molecules that share consistent structural features. 
In these experiments, we train on 200 states with a batch size of 100 and assess performance by computing the Wasserstein distance between all 488 generated samples and the target samples.
While there is no clear connection between the target distribution and the number of modes in multimodal $r(\bz)$, our results on the validation data, across trials with $M=1, 4, 8, 16$, suggest that $M=4$ is a reliable setting.

\begin{figure}[t]
  \centering
  \includegraphics[width=0.9\columnwidth]{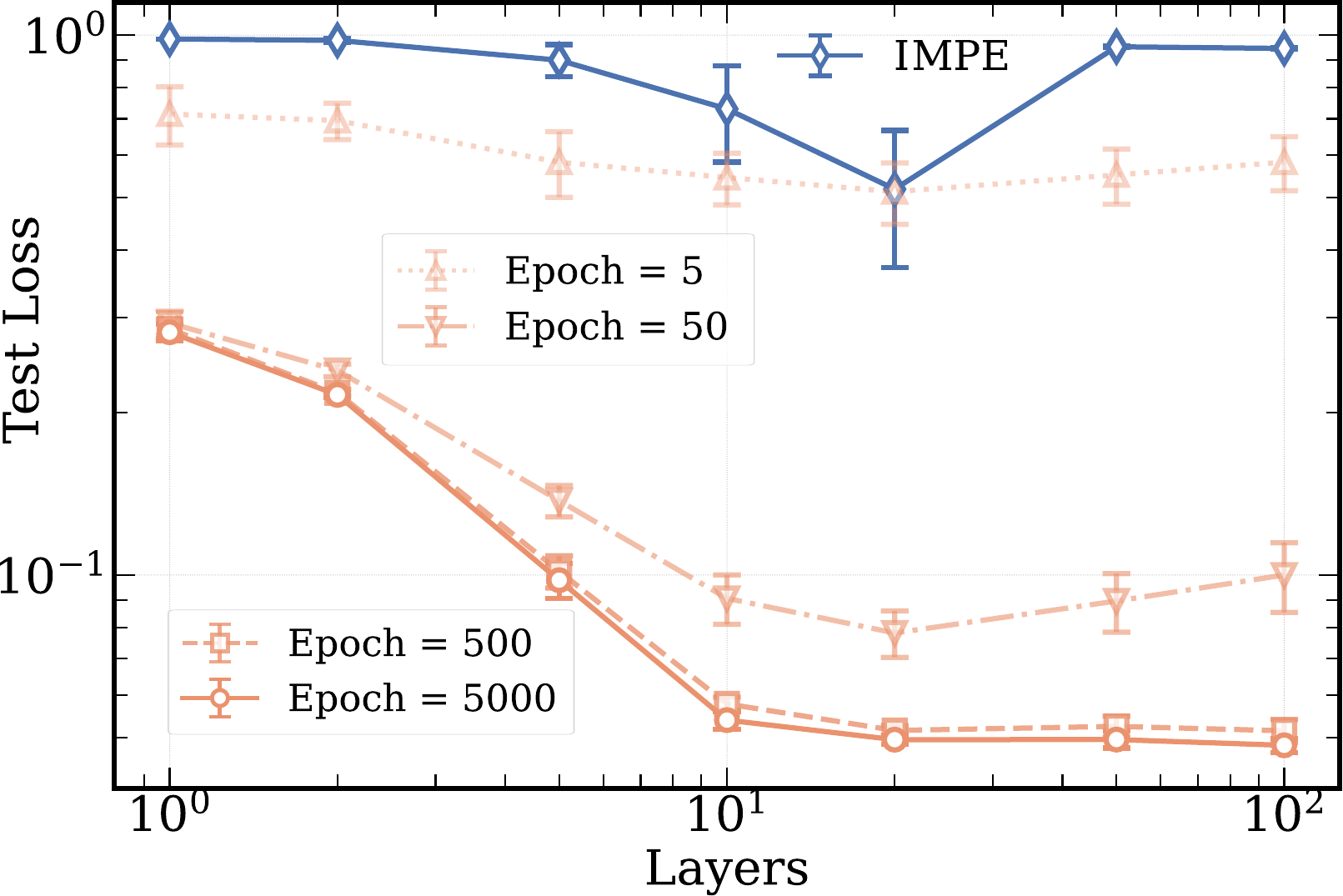}
  \caption{Comparison of LPQC ($M=4$, $E=1$) and IMPE~\cite{tran:2025:QUAP} on QM9-derived distributions, showing Wasserstein distance versus PQC layers $L$ and training epochs. Solid and dashed lines denote averages over 10 (LPQC) and 20 (IMPE) trials, respectively, with error bars indicating standard deviations.}
  \label{fig:compare:QM9}
\end{figure}

\begin{table*}[t]
\centering
\caption{Generation quality on the QM9-derived 3-D molecular dataset (8 heavy atoms, 2 ring systems; 4236 molecules in the dataset, 2000 used for training). ``Num.\ atoms'' is the average number of heavy atoms per generated molecule (target: 8.0). ``Wass'' is the Wasserstein distance to the target distribution of quantum states. The remaining columns are validity, uniqueness-times-validity, and novelty-times-validity, all expressed as percentages. LPQC achieves competitive generation quality compared to both classical and quantum baselines while preserving the correct mean atom count.}
\label{tab:QM9:generation}
\begin{ruledtabular}
\begin{tabular}{l c c c c c}
Method   & Num.\ atoms & Wass  & Valid (\%) & U $\times$ V (\%) & N $\times$ V (\%) \\
\hline
LPQC     & $8.0$ & $0.055$ & $49.86$ & $7.06$  & $21.26$ \\
LMLP     & $8.0$ & $0.039$ & $91.80$ & $4.79$  & $39.45$ \\
QVAE-Mole& $2.9$ & $0.181$ & $87.53$ & $6.71$  & $24.81$ \\
IMPE     & $7.9$ & $0.657$ & $24.74$ & $21.10$ & $21.76$ \\
QuDDPM   & $7.3$ & $0.570$ & $27.15$ & $16.34$ & $21.17$ \\
\end{tabular}
\end{ruledtabular}
\end{table*}

\begin{figure*}[t]
  \centering
  \includegraphics[width=2\columnwidth]{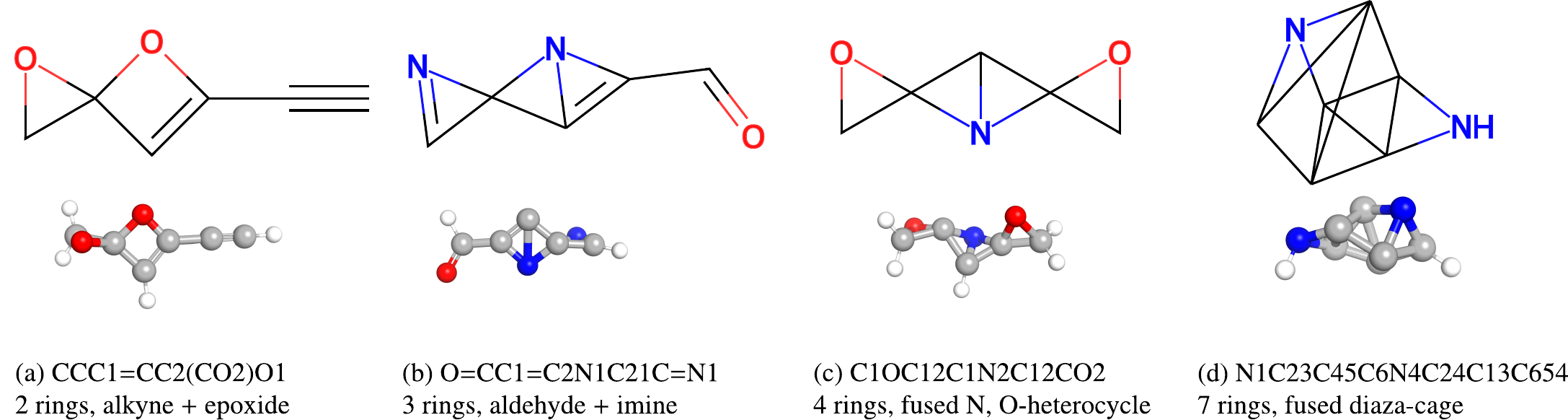}
  \caption{Four representative novel molecules generated by LPQC trained on the QM9-derived subset. All molecules contain exactly $8$ heavy atoms, matching the training constraint, and were absent from the training set. (a) Two-ring scaffold with an alkyne and an epoxide. (b) Three-ring structure with an aldehyde, an imine, and three heteroatoms (two N, one O). (c) Four-ring fused cage incorporating a carbonyl and an alkyne. (d) Seven-ring heavily fused diaza-cage representing the high-complexity end of the generated distribution. The structural variety across panels illustrates that LPQC explores diverse ring topologies and functional-group placements rather than collapsing onto a few canonical scaffolds. 2-D structures rendered via RDKit~\cite{rdkit} from the generated SMILES strings.}
  \label{fig:qm9:showcase}
\end{figure*}
Figure~\ref{fig:compare:QM9} compares the performance of LPQC (with the latent Gaussian prior of $M=4$ and $E=1$) against the IMPE framework (see Appendix~\ref{appx:IMPE} for details) in terms of test loss as varying PQC layers $L$ and training epochs. Solid and dashed lines denote averages over 10 trials for LPQC and 20 trials for IMPE (reproduced from Ref.~\cite{tran:2025:QUAP}), respectively. For fairness, we maintain identical settings for training data size and batch size.
LPQC consistently outperforms IMPE, achieving lower Wasserstein distances even with a few training epochs. This aligns with our observations: the latent space facilitates smoother navigation of the loss landscape, enabling efficient transformation of classical distributions to quantum ensembles.
Results in Appendix~\ref{appx:results} also demonstrate that MoEs can significantly improve the single-expert case.

To evaluate LPQC against a broader range of baselines on the concrete downstream task of 3-D molecule generation, we follow the protocol of Ref.~\cite{wu:qvae-mole:2024} and extend our QM9 evaluation to a larger subset of 4236 molecules with 8 heavy atoms and 2 ring systems, using 2000 molecules for training. After sampling quantum states from each generative model, we decode them into 3-D molecular structures via amplitude readout (Appendix~\ref{appx:encoding:3D}). We report four metrics: the average number of decoded heavy atoms per molecule (target: 8.0), the Wasserstein distance to the target state distribution, the validity rate (the fraction of generated molecules that satisfy basic chemical validity constraints under RDKit~\cite{rdkit}), and the composite uniqueness-times-validity and novelty-times-validity metrics standard in molecular generative modeling. We compare LPQC ($M=4$, $E=1$) against the classical LMLP baseline, the quantum variational autoencoder QVAE-Mole~\cite{wu:qvae-mole:2024}, the IMPE framework~\cite{tran:2025:QUAP}, and the recent quantum diffusion model QuDDPM~\cite{zhang:qddpm:prl:2024}. Results are summarized in Table~\ref{tab:QM9:generation}.

LPQC achieves competitive composite metrics, maintains a low Wasserstein distance to the target distribution, and preserves the correct mean number of heavy atoms (8.0). Its results are comparable to the classical LMLP baseline, while substantially outperforming the quantum baselines (QVAE-Mole, IMPE, QuDDPM) in Wasserstein distance and in matching the target atom count. We note that LMLP, although strong on raw generation metrics, is not a viable model for quantum data generation. It requires full tomography of every generated quantum state, which is only feasible at small scales and does not yield usable quantum states. LPQC, therefore, occupies a distinct regime in this quantum-to-quantum task.

\emph{Qualitative showcase of generated molecules.---}
Beyond the aggregate metrics in Table~\ref{tab:QM9:generation}, we inspect individual molecules that LPQC produces as novel and valid samples on this dataset. Figure~\ref{fig:qm9:showcase} displays four representative novel molecules generated by the trained LPQC, selected to span the range of structural complexity observed across the unique novel structures in the generated set. The chosen molecules vary from two-ring scaffolds with simple functional groups to heavily fused polycyclic cages with multiple embedded heteroatoms, while all preserve the target heavy-atom count of $8$. The presence of distinct functional groups (alkynes, carbonyls, imines, ethers) indicates that LPQC explores a structurally varied region of chemical space rather than collapsing onto a few prototypical scaffolds. 

A natural question is why the generated molecules exhibit a wide range of
ring counts (2--7), when the training subset was filtered to contain only
molecules with exactly 2 rings. The generated molecules lie outside the strict ring-count distribution of the
training data, indicating that the model explores topologies beyond simple
memorization or interpolation of training samples. The generative model learns over the full continuous
distribution of quantum states associated with the training molecules, not
over discrete graph structures directly, and the decoding step from quantum
state back to a molecular SMILES (Appendix~\ref{appx:decoding:3D}) admits
some flexibility in ring connectivity. The 7-ring cage in panel (d) represents the high-extrapolation
end of this behavior. Such heavily fused polycyclic scaffolds are absent
from the training subset and are chemically strained. We
view this extrapolation as a feature when the goal is exploring chemical
space beyond a curated training set. However, it is also a limitation when the
goal is faithful reproduction of a specific structural distribution.


\section{Related Works and Discussion}\label{sec:discussion}
UATs are originally established for classical NNs to demonstrate their capacity to approximate arbitrary continuous functions on compact subsets of Euclidean spaces~\cite{cybenko:1989:sigmoid,hornik:1991:apprx}. In quantum settings, these theorems reveal that PQCs can serve as powerful approximators for a broad spectrum of functions over classical data, and even achieve the optimal resource compared with classical NNs~\cite{tran:2021:prl:uap,yu:2024:NeurIPS:UAP}. This extends to generative modeling, where PQCs have been shown to approximate continuous multivariate probability distributions via expectation value sampling from measurements~\cite{barthe:2025:pqc:multivariate}, leveraging the Holevo bound to derive tight resource estimates. While these theorems affirm the roles of PQCs in classical data tasks such as regression and density estimation, they underscore the need for further generalization to purely quantum domains, such as distributions over density matrices.

Our paper introduces a generative framework where a low-dimensional latent variable $\bz \in \mathbb{R}^d$ (with $d$ typically small) is sampled from a prior distribution $r(\bz)$, mapped to PQC parameters $\btheta(\bz)$ via classical NNs $f(\bz;\bw)$, and then used to produce quantum states $\rho_{\bw}(\bz)$ via the PQC $U(\btheta(\bz))$. This pushforward approximates an arbitrary target distribution $Q$ over density operators in the 1-Wasserstein distance. While PQCs with a universal gate set can trivially approximate any single quantum state, modeling distributions over states requires learning a parameter distribution $p(\btheta)$ and converting its randomness into quantum outputs. Directly modeling a complex $p(\boldsymbol{\theta})$, especially with high-dimensional $\boldsymbol{\theta} \in \mathbb{R}^K$, is prone to the curse of dimensionality where sampling, density estimation, and optimization become inefficient.

In contrast, the latent $\bz$ is low-dimensional ($d \ll K$), making it easier to model sophisticated priors $r(\bz)$. The deterministic map $\bz \mapsto \btheta(\bz)$ (via NNs) then lifts the complexity to the high-dimensional $\btheta$, akin to classical normalizing flows: start with a simple low-dimensional base distribution and use a powerful mapping to generate complex outputs. Our framework leverages classical expressivity to mitigate barren plateaus by biasing initializations toward gradient-rich regions. Appendix~\ref{appx:BP-discussion} makes this precise, proving that the LPQC class admits barren-plateau-free initializations with polynomial-size generators, while a guarantee along the full training trajectory remains open. The latent $\bz$ decouples distributional complexity (modeled in low dimensions) from the nonlinear mapping to states, making the model more tractable for quantum scales.

Another prominent idea is to consider the randomness of the measurement process, which can be controlled to model quantum distributions within the MPE framework~\cite{choi2023preparing-d8c,cotler2023emergent-4eb}. Recent research uses this to demonstrate the universality of approximating pure-state distributions with a single many-body wave function, with practical layer-wise training to enhance optimization~\cite{tran:2025:QUAP}. MPE provides a compact, projection-based method optimized for pure-state ensembles with efficient incremental training. At the same time, LPQCs offer a more versatile, parameter-generation paradigm for mixed-state distributions, emphasizing classical expressivity to enhance universality and mitigate barren plateaus. MPE, while simpler in its core projection, may require more quantum resources for large ensembles due to repeated projections. Another key difference is the verifiability of generative models through repeated sampling. In LPQCs, it is easy to control the random seed to generate the latent variable $\bz$. At the same time, it requires overhead in quantum resources for post-selection to control random measurements in the MPE framework.

The LPQC construction used in our empirical experiments is straightforward and reflects the procedure used in the UAT proof. However, it can be made more flexible. For instance, replacing the deterministic NN mapping $\bz \mapsto \btheta(\bz)$ with normalizing flows~\cite{zou:2025:npj:flow:vqe} would allow learning more complex transformations from the latent space to the parameter space. It can potentially improve the approximation of $p(\btheta)$ while preserving efficient sampling and density estimation, and also avoid the need to set the number $M$ of prior modes heuristically. This could further mitigate barren plateaus by generating parameter distributions conditioned on quantum-system contexts (e.g., molecular descriptors), thereby enabling warm-starting and accelerating convergence.

\section*{Code and Data Availability}
Source code to reproduce all experiments is available in the GitHub
repository~\cite{tran:2026:github}. The QM9 dataset is publicly available
from Ref.~\cite{ramakrishnan:2014:QM9}.


\bibliographystyle{apsrev4-2}
\bibliography{main}

@article{peruzzo:2014:VQE,
  title   = {A variational eigenvalue solver on a photonic quantum processor},
  author  = {Peruzzo, Alberto and McClean, Jarrod and Shadbolt, Peter and Yung, Man-Hong and Zhou, Xiao-Qi and Love, Peter J. and Aspuru-Guzik, Al{\'a}n and O'Brien, Jeremy L.},
  journal = {Nat. Commun.},
  volume  = {5},
  pages   = {4213},
  year    = {2014},
  month   = jul,
  doi     = {10.1038/ncomms5213},
  url     = {https://doi.org/10.1038/ncomms5213}
}

@article{cerezo:2021:vqas:natrev,
  title   = {Variational quantum algorithms},
  author  = {Cerezo, M. and Arrasmith, Andrew and Babbush, Ryan and Benjamin, Simon C. and Endo, Suguru and Fujii, Keisuke and McClean, Jarrod R. and Mitarai, Kosuke and Yuan, Xiao and Cincio, Lukasz and Coles, Patrick J.},
  journal = {Nat. Rev. Phys.},
  volume  = {3},
  number  = {9},
  pages   = {625--644},
  year    = {2021},
  month   = aug,
  doi     = {10.1038/s42254-021-00348-9},
  url     = {https://doi.org/10.1038/s42254-021-00348-9}
}

@misc{farhi:2014:qaoa,
  title         = {A Quantum Approximate Optimization Algorithm},
  author        = {Farhi, Edward and Goldstone, Jeffrey and Gutmann, Sam},
  year          = {2014},
  eprint        = {1411.4028},
  archivePrefix = {arXiv},
  primaryClass  = {quant-ph},
  url           = {https://arxiv.org/abs/1411.4028}
}

@misc{yi:2025:geom:opt:Lie,
  title         = {Geometric Optimization on {L}ie Groups: A {L}ie-Theoretic Explanation of Barren Plateau Mitigation for Variational Quantum Algorithms},
  author        = {Yi, Zhehao and Bhadani, Rahul},
  year          = {2025},
  eprint        = {2512.02078},
  archivePrefix = {arXiv},
  primaryClass  = {quant-ph},
  url           = {https://arxiv.org/abs/2512.02078}
}

@article{lucas:pra:2022:avoidBP,
  title   = {Avoiding barren plateaus with classical deep neural networks},
  author  = {Friedrich, Lucas and Maziero, Jonas},
  journal = {Phys. Rev. A},
  volume  = {106},
  number  = {4},
  pages   = {042433},
  year    = {2022},
  month   = oct,
  doi     = {10.1103/PhysRevA.106.042433},
  url     = {https://doi.org/10.1103/PhysRevA.106.042433}
}

@article{zou:2025:npj:flow:vqe,
  title   = {Generative flow-based warm start of the variational quantum eigensolver},
  author  = {Zou, Hang and Rahm, Martin and Kockum, Anton Frisk and Olsson, Simon},
  journal = {npj Quantum Inf.},
  volume  = {12},
  number  = {1},
  pages   = {5},
  year    = {2026},
  doi     = {10.1038/s41534-025-01159-x},
  url     = {https://doi.org/10.1038/s41534-025-01159-x}
}

@misc{tran:2025:QUAP,
  title         = {Universality of Many-body Projected Ensemble for Learning Quantum Data Distribution},
  author        = {Tran, Quoc Hoan and Chinzei, Koki and Endo, Yasuhiro and Oshima, Hirotaka},
  year          = {2026},
  eprint        = {2601.18637},
  archivePrefix = {arXiv},
  primaryClass  = {quant-ph},
  url           = {https://arxiv.org/abs/2601.18637}
}

@article{tran:2021:prl:uap,
  title   = {Universal Approximation Property of Quantum Machine Learning Models in Quantum-Enhanced Feature Spaces},
  author  = {Goto, Takahiro and Tran, Quoc Hoan and Nakajima, Kohei},
  journal = {Phys. Rev. Lett.},
  volume  = {127},
  number  = {9},
  pages   = {090506},
  year    = {2021},
  month   = aug,
  doi     = {10.1103/PhysRevLett.127.090506},
  url     = {https://doi.org/10.1103/PhysRevLett.127.090506}
}

@inproceedings{yu:2024:NeurIPS:UAP,
  title         = {Non-asymptotic approximation error bounds of parameterized quantum circuits},
  author        = {Yu, Zhan and Chen, Qiuhao and Jiao, Yuling and Li, Yinan and Lu, Xiliang and Wang, Xin and Yang, Jerry Zhijian},
  booktitle     = {Adv. Neural Inf. Process. Syst. (NeurIPS)},
  volume        = {37},
  year          = {2024},
  eprint        = {2310.07528},
  archivePrefix = {arXiv},
  primaryClass  = {quant-ph},
  url           = {https://arxiv.org/abs/2310.07528}
}

@article{barthe:2025:pqc:multivariate,
  title   = {Parameterized quantum circuits as universal generative models for continuous multivariate distributions},
  author  = {Barthe, Alice and Grossi, Michele and Vallecorsa, Sofia and Tura, Jordi and Dunjko, Vedran},
  journal = {npj Quantum Inf.},
  volume  = {11},
  number  = {1},
  pages   = {121},
  year    = {2025},
  month   = jul,
  doi     = {10.1038/s41534-025-01064-3},
  url     = {https://doi.org/10.1038/s41534-025-01064-3}
}

@article{cybenko:1989:sigmoid,
  title   = {Approximation by superpositions of a sigmoidal function},
  author  = {Cybenko, George},
  journal = {Math. Control Signals Syst.},
  volume  = {2},
  number  = {4},
  pages   = {303--314},
  year    = {1989},
  month   = dec,
  doi     = {10.1007/BF02551274},
  url     = {https://doi.org/10.1007/BF02551274}
}

@article{hornik:1991:apprx,
  title   = {Approximation capabilities of multilayer feedforward networks},
  author  = {Hornik, Kurt},
  journal = {Neural Netw.},
  volume  = {4},
  number  = {2},
  pages   = {251--257},
  year    = {1991},
  doi     = {10.1016/0893-6080(91)90009-T},
  url     = {https://doi.org/10.1016/0893-6080(91)90009-T}
}

@article{clean:2018:natcom:barren,
  title   = {Barren plateaus in quantum neural network training landscapes},
  author  = {McClean, Jarrod R. and Boixo, Sergio and Smelyanskiy, Vadim N. and Babbush, Ryan and Neven, Hartmut},
  journal = {Nat. Commun.},
  volume  = {9},
  pages   = {4812},
  year    = {2018},
  doi     = {10.1038/s41467-018-07090-4},
  url     = {https://doi.org/10.1038/s41467-018-07090-4}
}

@misc{rdkit,
  title        = {{RDKit}: Open-source cheminformatics},
  howpublished = {\url{https://www.rdkit.org}},
  note         = {Version 2025.03.6, DOI: 10.5281/zenodo.591637}
}

@article{zhang:2023:tensorcircuit,
  title     = {{TensorCircuit}: a Quantum Software Framework for the {NISQ} Era},
  author    = {Zhang, Shi-Xin and Allcock, Jonathan and Wan, Zhou-Quan and Liu, Shuo and Sun, Jiace and Yu, Hao and Yang, Xing-Han and Qiu, Jiezhong and Ye, Zhaofeng and Chen, Yu-Qin and Lee, Chee-Kong and Zheng, Yi-Cong and Jian, Shao-Kai and Yao, Hong and Hsieh, Chang-Yu and Zhang, Shengyu},
  journal   = {Quantum},
  volume    = {7},
  pages     = {912},
  year      = {2023},
  month     = feb,
  doi       = {10.22331/q-2023-02-02-912},
  url       = {https://doi.org/10.22331/q-2023-02-02-912}
}

@software{jax:2018:github,
  title   = {{JAX}: composable transformations of {Python}+{NumPy} programs},
  author  = {Bradbury, James and Frostig, Roy and Hawkins, Peter and Johnson, Matthew James and Leary, Chris and Maclaurin, Dougal and Necula, George and Paszke, Adam and Vander{P}las, Jake and Wanderman-{M}ilne, Skye and Zhang, Qiao},
  version = {0.3.13},
  year    = {2018},
  url     = {http://github.com/jax-ml/jax}
}

@book{fremlin2011measure,
  title     = {Measure Theory: Topological Measure Spaces. Volume 4},
  author    = {Fremlin, D. H.},
  year      = {2011},
  publisher = {Torres Fremlin},
  url       = {https://www1.essex.ac.uk/maths/people/fremlin/mt.htm}
}

@article{jaros:2009:superfid,
  title   = {Sub- and super-fidelity as bounds for quantum fidelity},
  author  = {Miszczak, Jaros{\l}aw Adam and Pucha{\l}a, Zbigniew and Horodecki, Pawe{\l} and Uhlmann, Armin and {\.Z}yczkowski, Karol},
  journal = {Quantum Inf. Comput.},
  volume  = {9},
  number  = {1},
  pages   = {103--130},
  year    = {2009},
  month   = jan,
  issn    = {1533-7146},
  doi     = {10.26421/QIC9.1-2-7},
  url     = {https://doi.org/10.26421/QIC9.1-2-7}
}

@article{choi2023preparing-d8c,
  title   = {Preparing random states and benchmarking with many-body quantum chaos},
  author  = {Choi, Joonhee and Shaw, Adam L. and Madjarov, Ivaylo S. and Xie, Xin and Finkelstein, Ran and Covey, Jacob P. and Cotler, Jordan S. and Mark, Daniel K. and Huang, Hsin-Yuan and Kale, Anant and Pichler, Hannes and Brand{\~a}o, Fernando G. S. L. and Choi, Soonwon and Endres, Manuel},
  journal = {Nature},
  volume  = {613},
  number  = {7944},
  pages   = {468--473},
  year    = {2023},
  doi     = {10.1038/s41586-022-05442-1},
  url     = {https://doi.org/10.1038/s41586-022-05442-1}
}

@article{cotler2023emergent-4eb,
  title   = {Emergent Quantum State Designs from Individual Many-Body Wave Functions},
  author  = {Cotler, Jordan S. and Mark, Daniel K. and Huang, Hsin-Yuan and Hern{\'a}ndez, Felipe and Choi, Joonhee and Shaw, Adam L. and Endres, Manuel and Choi, Soonwon},
  journal = {PRX Quantum},
  volume  = {4},
  number  = {1},
  pages   = {010311},
  year    = {2023},
  doi     = {10.1103/PRXQuantum.4.010311},
  url     = {https://doi.org/10.1103/PRXQuantum.4.010311}
}

@article{ramakrishnan:2014:QM9,
  title   = {Quantum chemistry structures and properties of 134 kilo molecules},
  author  = {Ramakrishnan, Raghunathan and Dral, Pavlo O. and Rupp, Matthias and von Lilienfeld, O. Anatole},
  journal = {Sci. Data},
  volume  = {1},
  pages   = {140022},
  year    = {2014},
  month   = aug,
  doi     = {10.1038/sdata.2014.22},
  url     = {https://doi.org/10.1038/sdata.2014.22}
}

@misc{rathi:2023:3dqae:encoding,
  title         = {{3D-QAE}: Fully Quantum Auto-Encoding of {3D} Point Clouds},
  author        = {Rathi, Lakshika and Tretschk, Edith and Theobalt, Christian and Dabral, Rishabh and Golyanik, Vladislav},
  year          = {2023},
  eprint        = {2311.05604},
  archivePrefix = {arXiv},
  primaryClass  = {quant-ph},
  url           = {https://arxiv.org/abs/2311.05604}
}

@inproceedings{wu:qvae-mole:2024,
  title     = {{QVAE-Mole}: The Quantum {VAE} with Spherical Latent Variable Learning for {3-D} Molecule Generation},
  author    = {Wu, Huaijin and Ye, Xinyu and Yan, Junchi},
  booktitle = {Adv. Neural Inf. Process. Syst. (NeurIPS)},
  volume    = {37},
  year      = {2024},
  url       = {https://openreview.net/forum?id=RqvesBxqDo}
}

@article{jacobs:1991:moe,
  title   = {Adaptive Mixtures of Local Experts},
  author  = {Jacobs, Robert A. and Jordan, Michael I. and Nowlan, Steven J. and Hinton, Geoffrey E.},
  journal = {Neural Comput.},
  volume  = {3},
  number  = {1},
  pages   = {79--87},
  year    = {1991},
  doi     = {10.1162/neco.1991.3.1.79},
  url     = {https://doi.org/10.1162/neco.1991.3.1.79}
}

@article{jordan:1994:hierarchical:moe,
  title   = {Hierarchical Mixtures of Experts and the {EM} Algorithm},
  author  = {Jordan, Michael I. and Jacobs, Robert A.},
  journal = {Neural Comput.},
  volume  = {6},
  number  = {2},
  pages   = {181--214},
  year    = {1994},
  doi     = {10.1162/neco.1994.6.2.181},
  url     = {https://doi.org/10.1162/neco.1994.6.2.181}
}

@inproceedings{shazeer:2017:sparse:moe,
  title     = {Outrageously Large Neural Networks: The Sparsely-Gated Mixture-of-Experts Layer},
  author    = {Shazeer, Noam and Mirhoseini, Azalia and Maziarz, Krzysztof and Davis, Andy and Le, Quoc and Hinton, Geoffrey and Dean, Jeff},
  booktitle = {Int. Conf. Learn. Represent. (ICLR)},
  year      = {2017},
  url       = {https://openreview.net/forum?id=B1ckMDqlg}
}

@article{fedus:2022:switch,
  title   = {Switch Transformers: Scaling to Trillion Parameter Models with Simple and Efficient Sparsity},
  author  = {Fedus, William and Zoph, Barret and Shazeer, Noam},
  journal = {J. Mach. Learn. Res.},
  volume  = {23},
  number  = {120},
  pages   = {1--39},
  year    = {2022},
  url     = {http://jmlr.org/papers/v23/21-0998.html}
}

@article{rudolph:2022:wasserstein,
  title   = {Variational Quantum Generators: Generative Adversarial Quantum Machine Learning for Continuous Distributions},
  author  = {Romero, Jonathan and Aspuru-Guzik, Al{\'a}n},
  journal = {Adv. Quantum Technol.},
  volume  = {4},
  number  = {1},
  pages   = {2000003},
  year    = {2021},
  doi     = {10.1002/qute.202000003},
  url     = {https://doi.org/10.1002/qute.202000003}
}

@misc{qute:2020:hybrid:gen,
  title         = {Latent Style-based Quantum {GAN} for high-quality Image Generation},
  author        = {Chang, Su Yeon and Thanasilp, Supanut and Le Saux, Bertrand and Vallecorsa, Sofia and Grossi, Michele},
  year          = {2024},
  eprint        = {2406.02668},
  archivePrefix = {arXiv},
  primaryClass  = {quant-ph},
  url           = {https://arxiv.org/abs/2406.02668}
}

@article{zhang:qddpm:prl:2024,
  title   = {Generative Quantum Machine Learning via Denoising Diffusion Probabilistic Models},
  author  = {Zhang, Bingzhi and Xu, Peng and Chen, Xiaohui and Zhuang, Quntao},
  journal = {Phys. Rev. Lett.},
  volume  = {132},
  number  = {10},
  pages   = {100602},
  year    = {2024},
  month   = mar,
  doi     = {10.1103/PhysRevLett.132.100602},
  url     = {https://doi.org/10.1103/PhysRevLett.132.100602}
}

@article{abrams:1999:qpe,
  title   = {Quantum Algorithm Providing Exponential Speed Increase for Finding Eigenvalues and Eigenvectors},
  author  = {Abrams, Daniel S. and Lloyd, Seth},
  journal = {Phys. Rev. Lett.},
  volume  = {83},
  number  = {24},
  pages   = {5162--5165},
  year    = {1999},
  doi     = {10.1103/PhysRevLett.83.5162},
  url     = {https://doi.org/10.1103/PhysRevLett.83.5162}
}

@inproceedings{gilyen:2019:qsvt,
  title     = {Quantum Singular Value Transformation and Beyond: Exponential Improvements for Quantum Matrix Arithmetics},
  author    = {Gily{\'e}n, Andr{\'a}s and Su, Yuan and Low, Guang Hao and Wiebe, Nathan},
  booktitle = {Proc. ACM Symp. Theory Comput. (STOC)},
  pages     = {193--204},
  year      = {2019},
  doi       = {10.1145/3313276.3316366},
  url       = {https://doi.org/10.1145/3313276.3316366}
}

@article{ragone:2024:lie:bp,
  title   = {A {Lie} Algebraic Theory of Barren Plateaus for Deep Parameterized Quantum Circuits},
  author  = {Ragone, Michael and Bakalov, Bojko N. and Sauvage, Fr{\'e}d{\'e}ric and Kemper, Alexander F. and Ortiz Marrero, Carlos and Larocca, Mart{\'i}n and Cerezo, M.},
  journal = {Nat. Commun.},
  volume  = {15},
  number  = {1},
  pages   = {7172},
  year    = {2024},
  doi     = {10.1038/s41467-024-49909-3},
  url     = {https://doi.org/10.1038/s41467-024-49909-3}
}

@inproceedings{li:2018:measuring:intrinsic:dim,
  title         = {Measuring the Intrinsic Dimension of Objective Landscapes},
  author        = {Li, Chunyuan and Farkhoor, Heerad and Liu, Rosanne and Yosinski, Jason},
  booktitle     = {Int. Conf. Learn. Represent. (ICLR)},
  year          = {2018},
  eprint        = {1804.08838},
  archivePrefix = {arXiv},
  url           = {https://openreview.net/forum?id=ryup8-WCW}
}

@article{grant:2019:identity:block,
  title   = {An Initialization Strategy for Addressing Barren Plateaus in Parametrized Quantum Circuits},
  author  = {Grant, Edward and Wossnig, Leonard and Ostaszewski, Mateusz and Benedetti, Marcello},
  journal = {Quantum},
  volume  = {3},
  pages   = {214},
  year    = {2019},
  doi     = {10.22331/q-2019-12-09-214},
  url     = {https://doi.org/10.22331/q-2019-12-09-214}
}

@article{sack:2022:warm:bp,
  title   = {Avoiding Barren Plateaus Using Classical Shadows},
  author  = {Sack, Stefan H. and Medina, Raimel A. and Michailidis, Alexios A. and Kueng, Richard and Serbyn, Maksym},
  journal = {PRX Quantum},
  volume  = {3},
  number  = {2},
  pages   = {020365},
  year    = {2022},
  doi     = {10.1103/PRXQuantum.3.020365},
  url     = {https://doi.org/10.1103/PRXQuantum.3.020365}
}

@software{tran:2026:github,
  title   = {Source code of the paper ``{Latent}-Conditioned Parameterized Quantum Circuits as Universal Approximators for Distributions over Quantum States''},
  author  = {Tran, Quoc Hoan and Chinzei, Koki and Endo, Yasuhiro and Oshima, Hirotaka},
  year    = {2026},
  url     = {https://github.com/FujitsuResearch/latent-uap}
}

@article{terhal:2000:gibbs,
  title   = {Problem of equilibration and the computation of correlation functions on a quantum computer},
  author  = {Terhal, Barbara M. and DiVincenzo, David P.},
  journal = {Phys. Rev. A},
  volume  = {61},
  number  = {2},
  pages   = {022301},
  year    = {2000},
  month   = jan,
  doi     = {10.1103/PhysRevA.61.022301},
  url     = {https://doi.org/10.1103/PhysRevA.61.022301}
}

@article{poulin:2009:gibbs,
  title   = {Sampling from the Thermal Quantum Gibbs State and Evaluating Partition Functions with a Quantum Computer},
  author  = {Poulin, David and Wocjan, Pawel},
  journal = {Phys. Rev. Lett.},
  volume  = {103},
  number  = {22},
  pages   = {220502},
  year    = {2009},
  month   = nov,
  doi     = {10.1103/PhysRevLett.103.220502},
  url     = {https://doi.org/10.1103/PhysRevLett.103.220502}
}

@article{liu:2018:qad,
  title   = {Quantum machine learning for quantum anomaly detection},
  author  = {Liu, Nana and Rebentrost, Patrick},
  journal = {Phys. Rev. A},
  volume  = {97},
  number  = {4},
  pages   = {042315},
  year    = {2018},
  month   = apr,
  doi     = {10.1103/PhysRevA.97.042315},
  url     = {https://doi.org/10.1103/PhysRevA.97.042315}
}

@article{childs:2021:trotter:prx,
  title   = {Theory of {T}rotter Error with Commutator Scaling},
  author  = {Childs, Andrew M. and Su, Yuan and Tran, Minh C. and Wiebe, Nathan and Zhu, Shuchen},
  journal = {Phys. Rev. X},
  volume  = {11},
  number  = {1},
  pages   = {011020},
  year    = {2021},
  doi     = {10.1103/PhysRevX.11.011020},
  url     = {https://doi.org/10.1103/PhysRevX.11.011020}
}

@inproceedings{zhang:2022:gaussian:init,
  title     = {Escaping from the Barren Plateau via {G}aussian Initializations in Deep Variational Quantum Circuits},
  author    = {Zhang, Kaining and Liu, Liu and Hsieh, Min-Hsiu and Tao, Dacheng},
  booktitle = {Adv. Neural Inf. Process. Syst. (NeurIPS)},
  volume    = {35},
  pages     = {18612--18627},
  year      = {2022},
  eprint    = {2203.09376},
  archivePrefix = {arXiv},
  url       = {https://arxiv.org/abs/2203.09376}
}


\appendix

\section{Encoding Molecules to Quantum States}\label{appx:encoding:3D}

The QM9 dataset features small organic molecules, each incorporating up to
9 heavy atoms (C, N, O, F) supplemented by hydrogens, for a maximum of 29
atoms overall per molecule. For any given molecule (indexed as
$j = 0, 1, \ldots, N-1$ within the dataset), it is expressed as
$\{(\bv^j_i, \ba^j_i)\}_{i=0}^{m_j-1}$, with $m_j$ denoting the heavy-atom
count, $\bv^j_i \in \mathbb{R}^3$ representing the 3-D positions, and
$\ba^j_i \in \{0,1\}^k$ the one-hot encoded atomic species for the $j$th
molecule. We focus on heavy atoms only, giving $k=4$ distinct atom types
(C, N, O, F).

The encoding must address arbitrary shifts, orientations, atom permutations,
and quantum normalization, while ensuring that the resulting amplitudes are
real and non-negative for ease of both encoding and decoding. We use the
strategy of Refs.~\cite{wu:qvae-mole:2024,rathi:2023:3dqae:encoding}, with the following
key stages.

\begin{enumerate}
\item \textbf{Canonical atom ordering.} Atoms are reordered using canonical
    SMILES strings generated via the RDKit toolkit~\cite{rdkit}, producing
    an ordering independent of the original input.
\item \textbf{Centering and rotation.} The molecule is translated so its centroid coincides with the origin. It is then rotated such that the first atom from the canonical SMILES ordering aligns with the $z$-axis.
\item \textbf{Global bounding-box normalization.} Across the entire training
    subset, we compute the per-axis minimum
    $\bv_{\min} = (v_{\min,x}, v_{\min,y}, v_{\min,z})$ and the scalar
    bounding-box side
    $\Delta = \max_{a \in \{x,y,z\}}(v_{\max,a} - v_{\min,a})$.
    Each atom's coordinates are then normalized as
    $\tilde{\bv}_i = (\bv_i - \bv_{\min})/\Delta$, placing all normalized
    coordinates in $[0,1]$ along the longest axis and within $[0,1]$ along
    the other two.
\item \textbf{Per-atom feature concatenation.} For atom $i$ with normalized
    position $\tilde{\bv}_i = (\tilde x_i, \tilde y_i, \tilde z_i)$ and
    one-hot atom-type vector $\ba_i = (p_i^C, p_i^N, p_i^O, p_i^F)$,
    the per-atom feature vector is
    $(\tilde x_i, \tilde y_i, \tilde z_i, p_i^C, p_i^N, p_i^O, p_i^F)$,
    of length $3 + k = 7$.
\item \textbf{Auxiliary vector.} A length-$m$ auxiliary vector $\boldsymbol{\alpha}$
    is appended, with $i$-th entry
    $\alpha_i = \sqrt{\max(0,\; 3 - \tilde x_i^2 - \tilde y_i^2 - \tilde z_i^2)}$.
    This auxiliary serves to make the encoded vector reach unit norm after
    global rescaling and is not used in decoding.
\item \textbf{Global normalization and amplitude encoding.} The per-atom
    feature vectors are concatenated, followed by the auxiliary vector and
    zero-padding to length $2^n - 1$, where $n = 7$ is the qubit count.
    The final amplitude (index $2^n - 1$) stores the atom count $m$ as
    bookkeeping. The resulting vector is divided by $2\sqrt{m}$ to obtain
    unit Euclidean norm, and the amplitudes $\{\alpha_j\}_{j=0}^{2^n-1}$
    define the encoded state $\ket{\psi} = \sum_j \alpha_j \ket{j}$ on $n=7$
    qubits.
\end{enumerate}

\section{Decoding 3-D Molecules from Generated Quantum States and SMILES Generation}\label{appx:decoding:3D}

The encoding procedure in Appendix~\ref{appx:encoding:3D} maps each QM9
molecule to a 7-qubit pure quantum state whose amplitudes carry the
zero-padded, normalized concatenation of per-atom (position, type)
vectors. To evaluate generation quality, we apply the inverse procedure
to each quantum state produced by LPQC. The pipeline follows the
graph-theoretic protocol of
Refs.~\cite{wu:qvae-mole:2024,rathi:2023:3dqae:encoding}, and proceeds
through five stages.

\paragraph{Stage 1: Amplitude readout.}
LPQC generates density matrices on the 7-qubit Hilbert space. For the
QM9 task, the generated states are pure
($\rho_\bw(\bz) = \ket{\psi(\bz)}\!\bra{\psi(\bz)}$, with $m_{\text{anc}}=0$
ancilla qubits). The decoder operates on the amplitude vector
$\bx = |\psi(\bz)|$ obtained by taking the elementwise modulus of the
wavefunction amplitudes. Since the encoder produces non-negative real
amplitudes by construction, modulus is the natural inverse of the
decoding side and discards only the global phase, which has no physical
meaning.

\paragraph{Stage 2: From amplitudes to atomic positions and types.}
The amplitude vector $\bx \in \mathbb{R}^{2^7}$ is partitioned into a
sequence of $7$-entry per-atom blocks followed by an auxiliary block.
The $i$-th atom block occupies indices $[7i, 7i+7)$ in the layout
$[\tilde x_i, \tilde y_i, \tilde z_i, p_i^C, p_i^N, p_i^O, p_i^F]$. We
determine the number of occupied atom slots $m$ as follows: starting
from $i=0$, we check whether $\sum_{t \in \{C,N,O,F\}} p_i^t > \tau$
for an occupancy threshold $\tau$ (we use $\tau = 0.4$). The first $i$
for which this sum falls below $\tau$ is taken as $m$. Alternatively,
when the target atom count is known in advance, $m$ may be set directly
without thresholding.

To invert the encoder's global normalization, each amplitude is first
rescaled by a factor of $2\sqrt{m}$: $\tilde{\bx} = 2\sqrt{m}\,\bx$. The
atom type of the $i$-th occupied slot is assigned by
$t_i = \arg\max_{t \in \{C,N,O,F\}}\tilde p_i^t$,
corresponding to atomic numbers $\{6, 7, 8, 9\}$ respectively, and the
unnormalized coordinates are read out as
$\tilde{\bv}_i = (\tilde x_i, \tilde y_i, \tilde z_i)$. Physical
coordinates are then recovered via the affine map
$
    \bv_i \;=\; 2\Delta\,\tilde{\bv}_i \;+\; \bv_{\min},
$
where $\Delta$ is the scalar bounding-box side and $\bv_{\min}$ is the
per-axis minimum, both saved from the training set. For our QM9 subset,
$\Delta \approx 10.36$\,\AA{} and
$\bv_{\min} \approx (-3.93, -4.54, -5.11)$\,\AA{}. The factor of $2$
relative to the strict inverse of the encoder's normalization
($\bv_i = \Delta\tilde{\bv}_i + \bv_{\min}$) follows the convention of
Ref.~\cite{wu:qvae-mole:2024} and effectively rescales the coordinate
spread relative to the fixed covalent-bond thresholds used in Stage~3.
Empirically, this choice improves the validity rate of decoded molecules
at the cost of a modest reduction in uniqueness, reflecting that the
stretched coordinates produce sparser, more chemically conservative
adjacency matrices.

\paragraph{Stage 3: Bond inference via covalent-radius distance thresholding.}
Given the recovered Cartesian coordinates $\{\bv_i\}_{i=1}^m$ and atomic
numbers $\{a_i\}_{i=1}^m$ (with $a_i \in \{6,7,8,9\}$ for C, N, O, F), we
construct the adjacency-connectivity matrix
$\mathrm{AC} \in \{0,1\}^{m \times m}$ by thresholding pairwise interatomic
distances against tabulated covalent-bond cutoffs. Specifically, for each
atom pair $(i,j)$,
\[
    \mathrm{AC}_{ij} \;=\; \begin{cases}
        1 & \text{if } d_{ij} \le \tau_{a_i a_j} \text{ and } \deg(j) < \nu_{\max}(a_j), \\
        0 & \text{otherwise},
    \end{cases}
\]
where $d_{ij} = \|\bv_i - \bv_j\|$ and $\tau_{a_i a_j}$ is a covalent-bond
threshold (computed as the sum of the covalent radii of atoms $i$ and $j$
scaled by a factor of $2.2$ to accommodate slightly extended generated
geometries), and $\nu_{\max}(a_j)$ is the maximum valence of atom $j$
($4$ for C, $3$ for N, $2$ for O, $1$ for F). The valence constraint during
adjacency construction prevents over-bonding in regions of unusually high
atomic density. Connectivity is then checked to determine whether the graph is connected. 
A molecular graph that fragments into disconnected components is rejected as invalid.

\paragraph{Stage 4: Bond-order assignment via valence completion.}
The adjacency matrix $\mathrm{AC}$ only encodes \emph{whether} two atoms
are bonded, not the bond order (single, double, triple). To assign bond
orders, we solve the valence-completion problem: find a bond-order matrix
$\mathrm{BO} \in \mathbb{Z}_{\ge 0}^{m \times m}$ such that
$\mathrm{BO}_{ij} \ge \mathrm{AC}_{ij}$ for all $(i,j)$ and the row sums
$\sum_j \mathrm{BO}_{ij}$ match the standard valence of each atom type.
Algorithmically, we identify the set of unsaturated atoms
$\mathrm{UA} = \{i : \deg_{\mathrm{AC}}(i) < \nu(a_i)\}$ and greedily
promote single bonds within $\mathrm{UA}$ to double and triple bonds via
a maximum-weight matching on the unsaturated-atom subgraph, iterating
until all valences are satisfied or the procedure terminates without a
valid assignment. This step is implemented using the graph-theoretic
procedure of Refs.~\cite{wu:qvae-mole:2024,rathi:2023:3dqae:encoding}. 
If no valid bond-order assignment exists,
the molecule is flagged invalid.

\paragraph{Stage 5: Construction of RDKit molecule and SMILES generation.}
With atomic positions, types, and bond orders in hand, we construct an
RDKit \texttt{RWMol} object by sequentially adding atoms with their atomic
numbers and bonds with their inferred orders. The resulting molecule is
sanitized using RDKit's standard sanitization (\texttt{Chem.SanitizeMol}),
which checks aromaticity perception, valence consistency, kekulization, and stereochemistry assignment; molecules that fail sanitization are flagged invalid. The canonical SMILES string is then obtained via
\texttt{Chem.MolToSmiles(mol, isomericSmiles=True)}. We additionally verify
chemical validity by round-tripping the SMILES (parse $\to$ regenerate)
and re-sanitizing, which catches edge cases such as stereochemical
inconsistencies that pass the first sanitization but fail the round trip.

\paragraph{Generation-quality metrics.}
The validity, uniqueness, and novelty metrics reported in
Table~\ref{tab:QM9:generation} are computed as follows. \emph{Validity}
is the fraction of generated quantum states whose decoded SMILES passes
both the sanitization and round-trip checks. \emph{Uniqueness} is the
fraction of valid SMILES that are distinct as canonical strings.
\emph{Novelty} is the fraction of valid SMILES that do not appear in the
QM9 training set. The composite metrics $\mathrm{Unique} \times \mathrm{Valid}$ and $\mathrm{Novel} \times \mathrm{Valid}$ report the fraction of all generated
samples that are both valid and unique, or both valid and novel,
respectively, providing single measures that penalize models that produce
many invalid samples to inflate diversity. We use the implementation
provided by Ref.~\cite{wu:qvae-mole:2024} for all metric computations to
ensure consistency with prior work.

\section{Incremental Many-Body Projected Ensemble (IMPE)}\label{appx:IMPE}
We follow the explanation and experimental settings in \cite{tran:2025:QUAP} to explain the IMPE method.
The goal is to approximate a target distribution $Q_t$ over pure $n$-qubit states by constructing a parameterized ensemble $Q_{\bzet}$ through $T$ iterative cycles of unitary transformations and measurements. First, sample a training dataset $\mathcal{S} = \{ |\psi_0\rangle, \ldots, |\psi_{N_s-1}\rangle \}$ of $N_s$ states from $Q_t$. Start with an initial ensemble $\tilde{\mathcal{S}}_0 = \{ |\tilde{\psi}^{(0)}_j\rangle \}_j$ sampled from a simple random distribution, such as Haar product states.
At each cycle $t = 0, \ldots, T-1$, apply a PQC $V_t = V(\bzet_t)$ to the composite system comprising the primary data subsystem $\cH_S$ ($n_d$ qubits) and auxiliary subsystem $\cH_A$ ($n_a$ qubits), with states initialized as $|\tilde{\Psi}^{(t)}_j\rangle = |\tilde{\psi}^{(t)}_j\rangle_P \otimes |0\rangle_A$. Perform a projective measurement on $A$ in the computational basis, yielding outcome $\bz^{(t)}_j \in \{0,1\}^{n_a}$ and updating the data state via:
\begin{align}
\Phi^{(t)}_j(|\tilde{\psi}^{(t)}_j\rangle) &= \frac{(I_S \otimes \Pi_A) V_t |\tilde{\Psi}^{(t)}_j\rangle}{\sqrt{\langle \tilde{\Psi}^{(t)}_j | V_t^\dagger (I_S \otimes \Pi_A) V_t | \tilde{\Psi}^{(t)}_j \rangle}} \\
&= |\tilde{\psi}^{(t+1)}_j\rangle_S \otimes |\bz^{(t)}_j\rangle_A,
\end{align}
where $\Pi_A = |\bz^{(t)}_j\rangle\langle \bz^{(t)}_j|_A$. The resulting ensemble $\tilde{\mathcal{S}}_{t+1} = \{ |\tilde{\psi}^{(t+1)}_j\rangle \}_j$ iteratively refines the approximation, reducing resource demands compared to standard MPE.
The parameters $\bzet_t$ of each $V_t$ are optimized layer-wise to minimize a distance metric $\cD_{\mathrm{Wass}}(\mathcal{S}, \tilde{\mathcal{S}}_{t+1})$ between the training and generated ensembles; once optimized, $\bzet_t$ is fixed before proceeding to the next cycle. This decomposes the problem into $T$ sub-tasks, each with fewer parameters, thereby aiding convergence.
In numerical experiments, each $V_t$ employs a Hardware-Efficient Ansatz on $n_q=n_d + n_a$ qubits with $L$ layers:
\begin{equation*}
V_t(\bzet_t) = \prod_{l=1}^L \tilde{\Omega}_t \tilde{W}_t(\bzet_t),
\end{equation*}
where $\tilde{W}_t(\bzet_t) = \prod_{j=1}^{n_q} e^{-i\zeta_{t,2j-1}\frac{X_j}{2}} e^{-i\zeta_{t,2j-2}\frac{Y_j}{2}}$ implements single-qubit rotations (Pauli-$X$ and -$Y$ on qubit $j$, parameterized by $\zeta_{t,2j-1}$ and $\zeta_{t,2j-2}$), and $\tilde{\Omega}_t = \prod_{j=1}^{n_q-1} \mathrm{CZ}_{j,j+1}$ applies controlled-Z gates for nearest-neighbor entanglement.

In the experiments with the QM9-derived quantum dataset, we use the same experimental setting as IMPE and the dataset described in \cite{tran:2025:QUAP}. Here, $n_d=7$, $n_a=3$, $L$ varies in $\{1, 2, 5, 10, 20, 50, 100\}$ with $T=\frac{200}{L}$ cycles. Mini-batch training is employed with a batch size of $100$ for each $V_t$, using $200$ training samples over $100 \times L$ epochs.
Then the total of training epochs over $T$ cycles is $100\times L \times T = 2\times 10^4$ epochs.

\section{Lie-algebraic perspective on why the latent prior alleviates barren plateaus}\label{appx:BP-discussion}

The empirical observation in Sec.~\ref{sec:demo} that the latent-conditioned parameterization maintains gradient norms one to two orders of magnitude larger than the no-latent and RD baselines (Fig.~\ref{fig:gradnorm:qubits}) is consistent with a growing body of theoretical and empirical work on the origins of barren plateaus. We discuss the conjectural mechanism and leave a rigorous analysis as an open problem.

The recent unified analysis in Ref.~\cite{ragone:2024:lie:bp} expresses the gradient variance of a parameterized unitary $U(\btheta) = \prod_k e^{-i\theta_k G_k}$ as a sum of projections onto irreducible representations of the dynamical Lie algebra (DLA) $\mathfrak{g}$ generated by $\{iG_k\}$. When $\dim\mathfrak{g}$ scales as $\Theta(4^n)$ and the initial state and observable have nontrivial overlap with many irreducible blocks, the projections cancel via representation-theoretic orthogonality and the gradient variance decays as $O(e^{-cn})$. The standard no-latent baseline in which $\btheta$ is drawn from a broad distribution (Haar-like or independent uniform) on $\mathbb{R}^K$ is in precisely this regime. Here, a high-dimensional average over a large DLA produces the exponential suppression we observe in the leftmost column of Fig.~\ref{fig:gradnorm:qubits}.

The LPQC parameterization replaces this broad average with the pushforward measure $p_\bw = (f_\bw)_\#\, r$ on $\mathbb{R}^K$, induced by the latent prior $r(\bz)$ on the low-dimensional space $\cZ \subset \mathbb{R}^d$ ($d \ll K$). The expected squared gradient norm becomes
\begin{equation}\label{eq:pushforward}
    \mathbb{E}_{\bz \sim r}\bigl[\|\nabla_\btheta \cL(f_\bw(\bz))\|^2\bigr] \;=\; \int_\cZ \|\nabla_\btheta \cL(f_\bw(\bz))\|^2\, r(\bz)\,d\bz,
\end{equation}
which is a $d$-dimensional integral over the image manifold $\cM_\bw = \{f_\bw(\bz) : \bz \in \cZ\}$ rather than a $K$-dimensional integral over the full parameter space. This is the same kind of dimensionality reduction at work in classical settings where overparameterized neural networks are effectively trained on a low-dimensional manifold of parameter space~\cite{li:2018:measuring:intrinsic:dim}. The hypothesis underlying the LPQC approach is that the high-dimensional averaging used in the exponential suppression analysis does not apply once the support of the parameter distribution is constrained to a manifold. The cancellation across irreducible blocks of the DLA depends on the parameter distribution being spread sufficiently broadly to average over them.

Several lines of work support this intuition. It is empirically shown in Ref.~\cite{lucas:pra:2022:avoidBP} that NN-generated parameters exhibit larger gradient norms than uniformly sampled parameters across a range of ans\"atze and observables. In Ref.~\cite{yi:2025:geom:opt:Lie}, a complementary geometric perspective is provided. Parameters produced by smooth NN maps follow trajectories on a low-dimensional submanifold of the parameter space whose geometry is determined by the network's Jacobian, and these trajectories tend to avoid the flat regions characterized by the DLA analysis. Prior work on initialization strategies~\cite{grant:2019:identity:block} and classical-shadow-based warm-starting~\cite{sack:2022:warm:bp} demonstrates that careful choice of the initial parameter distribution can substantially reduce the severity of barren plateaus in practice.

The pushforward viewpoint above is a conjectural mechanism for the behavior of trained LPQCs. To upgrade it to a fully rigorous theory, one needs three ingredients: 
\begin{enumerate}
\item Identify low-dimensional submanifolds $\cM_\bw \subset \mathbb{R}^K$ that admit a uniform polynomial lower bound on $\|\nabla_\btheta \cL\|^2$
\item Show that such submanifolds are reachable by the image of a classical NN with polynomial parameter count
\item Verify that the meta-optimization over the generator weights $\bw$ itself does not introduce new barren plateaus
\end{enumerate}
We now show that (i) and (ii) admit affirmative, quantitative answers, and that (iii) admits an affirmative answer at initialization; what remains open is a guarantee along the entire training trajectory, which we discuss at the end of this appendix.

\subsection{Smoothness of PQC loss landscapes}

We first record a standard smoothness estimate. 

\begin{lemma}[Lipschitz gradients of expectation-value losses]\label{lem:hessian}
Let $U(\btheta) = \prod_{k=1}^{K} e^{-i\theta_k G_k} W_k$ with fixed unitaries $W_k$ and Hermitian generators $\norm{G_k}\le 1$, let $\rho(\btheta)$ denote the (possibly reduced) output state on input $\ket{0^{n+m}}$, and let $\cL_O(\btheta) = \Tr[O\,\rho(\btheta)]$ with $\norm{O}\le 1$. Then $|\partial_j \cL_O| \le 2$, $|\partial_j\partial_k \cL_O| \le 4$ for all $j,k$, and $\nabla_\btheta\cL_O$ is $\Lambda$-Lipschitz with $\Lambda \le 4K$.
\end{lemma}

\begin{proof}
We bound the first two derivatives of $\cL_O$ uniformly and then pass from the resulting
Hessian bound to Lipschitz continuity of the gradient. The single mechanism behind the
derivative bounds is that each parameter enters $U(\btheta)$ through one Hermitian generator,
so each $\partial_\theta$ inserts a single commutator with a (unitarily conjugated, hence
norm $\le 1$) generator, and a commutator at most doubles the operator norm:
\begin{equation}\label{eq:commutator-norm}
    \norm{[A,B]}=\norm{AB-BA}\le \norm{AB}+\norm{BA}\le 2\norm{A}\,\norm{B}.
\end{equation}

Tracing out the ancilla register $\cH_A$ is equivalent to evaluating an observable that acts
as the identity on $\cH_A$: with $\ket{\psi(\btheta)}:=U(\btheta)\ket{0^{n+m}}$ and
$\widetilde O:=O\otimes I_A$,
\begin{align}
    \cL_O(\btheta)&=\Tr\!\bigl[O\,\rho(\btheta)\bigr]\nonumber\\
    &=\Tr\!\bigl[(O\otimes I_A)\,\ket{\psi(\btheta)}\!\bra{\psi(\btheta)}\bigr]\nonumber\\
    &=\bra{\psi(\btheta)}\widetilde O\ket{\psi(\btheta)},
\end{align}
with $\norm{\widetilde O}=\norm{O}\le 1$. The partial trace therefore changes neither the
observable norm nor the bounds below, and we may work with
$\cL_O(\btheta)=\bra{0^{n+m}}U(\btheta)^\dagger\,\widetilde O\,U(\btheta)\ket{0^{n+m}}$.

\emph{Single-parameter form.}
Fix an index $j$ and hold all other coordinates fixed. Splitting the circuit immediately
around the $j$-th rotation, $U=U_{>j}\,e^{-i\theta_j G_j}\,U_{<j}$, set
\begin{equation*}
    \ket{\varphi}:=U_{<j}\ket{0^{n+m}},\qquad B:=U_{>j}^\dagger\,\widetilde O\,U_{>j}.
\end{equation*}
Then $\ket{\varphi}$ is a unit vector and $B$ is Hermitian with $\norm{B}=\norm{\widetilde O}\le 1$,
since conjugation by a unitary preserves the operator norm. As a function of $\theta_j$ alone,
\begin{equation}
    \cL_O=g(\theta_j):=\bra{\varphi}\,B(\theta_j)\,\ket{\varphi},
    \qquad B(\theta_j):=e^{i\theta_j G_j}\,B\,e^{-i\theta_j G_j},
\end{equation}
and differentiating the conjugation gives
$\partial_{\theta_j}B(\theta_j)=e^{i\theta_j G_j}\,i[G_j,B]\,e^{-i\theta_j G_j}$.

\emph{First derivatives.}
We express
\begin{equation}
    \partial_j\cL_O=\bra{\varphi}\,e^{i\theta_j G_j}\,i[G_j,B]\,e^{-i\theta_j G_j}\,\ket{\varphi}.
\end{equation}
The operator $i[G_j,B]$ is Hermitian, so its expectation in a unit vector is bounded by its
operator norm (by \eqref{eq:commutator-norm} and $\norm{G_j}\le 1$),
\begin{equation}
    |\partial_j\cL_O|\le \norm{[G_j,B]}\le 2\norm{G_j}\,\norm{B}\le 2.
\end{equation}

\emph{Second derivatives.}
For the diagonal entry, differentiating $B(\theta_j)$ twice inserts the commutator twice,
\begin{equation}
    \partial_j^2\cL_O
    =\bra{\varphi}\,e^{i\theta_j G_j}\,i^2[G_j,[G_j,B]]\,e^{-i\theta_j G_j}\,\ket{\varphi},
\end{equation}
with $\norm{[G_j,[G_j,B]]}\le 4\,\norm{G_j}^2\norm{B}\le 4$ so $|\partial_j^2\cL_O|\le 4$. For an off-diagonal entry $k\ne j$, write
$\partial_j\cL_O=\bra{\varphi}M\ket{\varphi}$ with the Hermitian operator
$M:=e^{i\theta_j G_j}\,i[G_j,B]\,e^{-i\theta_j G_j}$, $\norm{M}\le 2$. The coordinate $\theta_k$
enters exactly one of $\ket{\varphi}$ or $B$ (never both, as $k\ne j$):
\begin{itemize}
    \item If $k>j$, then $\theta_k$ enters through $B$; with $\widehat G_k$ denoting $G_k$
    conjugated by the gates that follow the $k$-th rotation (a unitary, so $\norm{\widehat G_k}\le 1$),
    one has $\partial_k B=U_{>j}^\dagger\,i[\widehat G_k,\widetilde O]\,U_{>j}$, whence
    $\norm{\partial_k B}\le 2$ and
    \begin{align}
        |\partial_k\partial_j\cL_O|
        &=\bigl|\bra{\varphi}\,e^{i\theta_j G_j}\,i[G_j,\partial_k B]\,e^{-i\theta_j G_j}\,\ket{\varphi}\bigr|\nonumber\\
        &\le \norm{[G_j,\partial_k B]}\le 2\norm{\partial_k B}\le 4
    \end{align}
    \item If $k<j$, then $\theta_k$ enters through $\ket{\varphi}$; with $\overline G_k$ denoting
    $G_k$ conjugated by the gates between the $k$-th and $j$-th rotations ($\norm{\overline G_k}\le 1$),
    one has $\partial_k\ket{\varphi}=-i\,\overline G_k\ket{\varphi}$, so
    \begin{align}
        &|\partial_k\partial_j\cL_O| =| \bra{\varphi}\,i[\overline G_k,M]\,\ket{\varphi} | \nonumber\\
        &\le |\partial_k\partial_j\cL_O|\le \norm{[\overline G_k,M]}\le 2\norm{\overline G_k}\norm{M}\le 4.
    \end{align}
\end{itemize}
In all cases $|\partial_j\partial_k\cL_O|\le 4$. (The same insertion argument gives, more generally,
$|\partial_{j_1}\!\cdots\partial_{j_m}\cL_O|\le 2^{m}\norm{\widetilde O}$.)

\emph{Hessian spectral norm.}
The Hessian $\nabla^2\cL_O$ is a real symmetric $K\times K$ matrix with
$|(\nabla^2\cL_O)_{jk}|=|\partial_j\partial_k\cL_O|\le 4$. By Gershgorin's circle theorem, every
eigenvalue $\lambda$ satisfies
$|\lambda-(\nabla^2\cL_O)_{jj}|\le\sum_{k\ne j}|(\nabla^2\cL_O)_{jk}|$ for some $j$, hence
\begin{equation}
    |\lambda|\le \sum_{k=1}^{K}|\partial_j\partial_k\cL_O|\le 4K.
\end{equation}
Therefore, $\norm{\nabla^2\cL_O(\btheta)}_2\le 4K=:\Lambda$ for every $\btheta$ (equivalently, the spectral norm of a symmetric matrix is at most its largest absolute row sum).

\emph{Lipschitz gradient.}
The uniform Hessian bound yields a Lipschitz gradient by the fundamental theorem of
calculus along the segment $\btheta_t:=\btheta'+t(\btheta-\btheta')$,
\begin{equation*}
    \nabla\cL_O(\btheta)-\nabla\cL_O(\btheta')
    =\int_0^1 \nabla^2\cL_O(\btheta_t)\,(\btheta-\btheta')\,dt,
\end{equation*}
so that
$\norm{\nabla\cL_O(\btheta)-\nabla\cL_O(\btheta')}
\le\bigl(\sup_{t\in[0,1]}\norm{\nabla^2\cL_O(\btheta_t)}_2\bigr)\,\norm{\btheta-\btheta'}
\le \Lambda\,\norm{\btheta-\btheta'}$ with $\Lambda\le 4K$.
\end{proof}

\begin{remark}
The training loss used in this work is the optimal-transport objective with superfidelity kernel. For pure output states ($m=0$, as in the QM9 experiments), the superfidelity reduces to $\kappa(\rho,\sigma) = \Tr(\rho\sigma)$, an expectation-value functional, and Lemma~\ref{lem:hessian} applies directly to each cost entry; the linear-program value is a $1$-Lipschitz, piecewise-linear function of the cost entries, so the composite loss inherits a Lipschitz gradient bound $\Lambda = O(K)$ wherever the optimal plan is locally constant. For $m>0$, the term $\sqrt{[1-\Tr\rho^2][1-\Tr\sigma^2]}$ has unbounded derivative as $\Tr\rho^2 \to 1$; the bound then holds with $\Lambda$ depending additionally on a lower bound on $1-\Tr\rho(\btheta)^2$ over the relevant parameter region. For clarity, the results below are stated for expectation-value losses.
\end{remark}

\subsection{Gradient-rich submanifolds are realizable by small generators}

The next proposition formalizes the key structural point: once a single gradient-rich parameter point exists, the LPQC family contains low-dimensional pushforwards---realized by generator networks of size $O(dK)$---on which the expected squared gradient norm is uniformly bounded below. This answers (i) and (ii) above, conditionally on the existence of such a point.

\begin{proposition}[Gradient preservation under concentrated pushforwards]\label{prop:concentrated}
Let $\cL:\mathbb{R}^K\to\mathbb{R}$ have $\Lambda$-Lipschitz gradient and suppose there exists $\btheta^\ast$ with $\norm{\nabla\cL(\btheta^\ast)} \ge g > 0$. Consider any generator class that contains, for every $\btheta_0\in\mathbb{R}^K$ and $\gamma\ge 0$, a map of the form $f_\bw(\bz) = \btheta_0 + \gamma\, s(\bz)$ with $\sup_{\bz\in\cZ}\norm{s(\bz)}\le 1$; e.g., a one-hidden-layer MLP with bounded activation, output bias $\btheta_0$, and output-weight norm $\gamma$, of parameter count $O(dK)$. Then, choosing $\btheta_0 = \btheta^\ast$ and any $\gamma \le g/(2\Lambda)$,
\begin{equation}\label{eq:concentrated-bound}
    \mathbb{E}_{\bz\sim r}\bigl[\norm{\nabla_\btheta\cL(f_\bw(\bz))}^2\bigr] \;\ge\; (g - \Lambda\gamma)^2 \;\ge\; \frac{g^2}{4}
\end{equation}
for every prior $r$ on $\cZ$. In particular, the image manifold $\cM_\bw = f_\bw(\cZ) \subseteq B(\btheta^\ast,\gamma)$ is an (at most) $d$-dimensional set on which $\norm{\nabla\cL}$ is uniformly lower bounded by $g/2$.
\end{proposition}

\begin{proof}
For every $\bz$, the Lipschitz property gives $\norm{\nabla\cL(f_\bw(\bz))} \ge \norm{\nabla\cL(\btheta^\ast)} - \Lambda\norm{f_\bw(\bz)-\btheta^\ast} \ge g - \Lambda\gamma \ge g/2$. Squaring and taking expectations yields Eq.~\eqref{eq:concentrated-bound}.
\end{proof}

Proposition~\ref{prop:concentrated} reduces questions (i)--(ii) to the existence of gradient-rich points $\btheta^\ast$, which is precisely what the literature on structured initializations provides. The sharpest such guarantee is due to Ref.~\cite{zhang:2022:gaussian:init}, which proves that for hardware-efficient circuits of depth $L$ with parameters drawn i.i.d.\ from $\cN(0,\gamma_0^2)$ with $\gamma_0^2 = O(1/L)$, the expected squared gradient norm of expectation-value losses decays at most polynomially in $n$ and $L$ (for local observables, and for global observables under stated conditions). Combining the two results yields an unconditional statement.

\begin{corollary}[LPQCs admit barren-plateau-free initializations]\label{cor:bp-free}
Consider the HEA of Eq.~\eqref{eqn:HEA} with $K = 2L(n+m)$ parameters and an expectation-value loss $\cL_O$, in a setting where the Gaussian-initialization guarantee of Ref.~\cite{zhang:2022:gaussian:init} applies, i.e., $\mathbb{E}_{\btheta\sim\cN(0,\gamma_0^2 I_K)}\norm{\nabla\cL_O}^2 \ge 1/\mathrm{poly}(n,L)$. Then:
\begin{enumerate}
    \item There exists $\btheta^\ast$ with $g^2 := \norm{\nabla\cL_O(\btheta^\ast)}^2 \ge 1/\mathrm{poly}(n,L)$.
    \item By Lemma~\ref{lem:hessian}, $\Lambda \le 4K = \mathrm{poly}(n,L)$, so the radius $\gamma = g/(2\Lambda)$ in Proposition~\ref{prop:concentrated} is inverse-polynomial, not exponentially small.
    \item Consequently, the LPQC family contains weight settings $\bw$, with a classical generator of $O(dK) = \mathrm{poly}(n,L)$ parameters, such that
    \begin{equation*}
        \mathbb{E}_{\bz\sim r}\bigl[\norm{\nabla_\btheta\cL_O(f_\bw(\bz))}^2\bigr] \ \ge\ \frac{1}{\mathrm{poly}(n,L)}.
    \end{equation*}
\end{enumerate}
In these settings, questions (i) and (ii) are therefore answered affirmatively: gradient-rich $d$-dimensional submanifolds exist and are reachable by polynomial-size generators.
\end{corollary}

\begin{remark}[Exact correspondence with Gaussian initialization]
When $d \ge K$, the Gauss-linear LPQC variant of Sec.~\ref{sec:demo}, $\btheta = \bW\bz + \bb$ with $\bz\sim\cN(\boldsymbol{0},\mathbf{I}_d)$, realizes the parameter distribution $\cN(\bb, \bW\bW^\top)$ exactly; choosing $\bW\bW^\top = \gamma_0^2\mathbf{I}_K$ reproduces verbatim the initialization analyzed in Ref.~\cite{zhang:2022:gaussian:init}, so its ensemble-average lower bound transfers directly to the LPQC. For $d < K$ (the regime of practical interest), the pushforward is supported on a $d$-dimensional affine subspace and Corollary~\ref{cor:bp-free} provides the corresponding guarantee via the concentration argument instead.
\end{remark}

\subsection{The meta-landscape at initialization}

Question (iii) asks whether optimizing the generator weights $\bw$ reintroduces a plateau. The chain rule gives $\nabla_\bw \mathbb{E}_{\bz}[\cL(f_\bw(\bz))] = \mathbb{E}_{\bz}[J_\bw(\bz)^\top \nabla_\btheta\cL(f_\bw(\bz))]$, where $J_\bw$ is the Jacobian of $f_\bw$ with respect to $\bw$; gradients in $\bw$ can therefore vanish only if $\nabla_\btheta\cL$ vanishes on the image manifold or if the Jacobian directions are orthogonal to it. For the output-bias block of any standard MLP the Jacobian is the identity, which yields the following unconditional lower bound at the initialization of Proposition~\ref{prop:concentrated}.

\begin{proposition}[No barren plateau in $\bw$ at initialization]\label{prop:meta}
In the setting of Proposition~\ref{prop:concentrated}, let the generator's final layer be $\btheta = \bW_{\mathrm{out}}\bh(\bz) + \bb$ with bias $\bb$. Then at the weights constructed there,
\begin{equation*}
    \norm{\nabla_\bw\,\mathbb{E}_{\bz}[\cL(f_\bw(\bz))]} \;\ge\; \norm{\nabla_{\bb}\,\mathbb{E}_{\bz}[\cL(f_\bw(\bz))]} \;\ge\; g - \Lambda\gamma \;\ge\; \frac{g}{2}.
\end{equation*}
\end{proposition}

\begin{proof}
$\nabla_{\bb}\,\mathbb{E}[\cL(f_\bw(\bz))] = \mathbb{E}[\nabla_\btheta\cL(f_\bw(\bz))]$, and
$\norm{\mathbb{E}[\nabla\cL(f_\bw(\bz))]} \ge \norm{\nabla\cL(\btheta^\ast)} - \mathbb{E}\norm{\nabla\cL(f_\bw(\bz)) - \nabla\cL(\btheta^\ast)} \ge g - \Lambda\gamma$
by the triangle inequality, Jensen's inequality, and the Lipschitz bound.
\end{proof}

Under the hypotheses of Corollary~\ref{cor:bp-free}, the right-hand side is $1/\mathrm{poly}(n,L)$, so the meta-gradient is estimable with polynomially many measurement shots: the meta-optimization over $\bw$ does not face a barren plateau at this initialization.

\subsection{Prior-induced coverage in the spread regime}\label{subsec:prior-coverage}
 
The constructions of Proposition~\ref{prop:concentrated} and Corollary~\ref{cor:bp-free} collapse the parameter pushforward into a ball $B(\btheta^\ast,\gamma)$, on which the latent prior plays no role.
Here, the bound of Proposition~\ref{prop:concentrated} holds for every prior $r$, the generator's output bias alone doing the work. The regime of our practical interest is to represent a genuine distribution over states rather than a point mass. We now give a complementary guarantee valid in this spread regime, in which the latent distribution becomes the active design variable. Throughout, fix an expectation-value loss $\cL_O$ as in Lemma~\ref{lem:hessian}, write $p_\btheta:=(f_\bw)_\#\, r$ for the parameter pushforward, and for a threshold $c>0$ define the gradient-rich set
\begin{equation}
    \cS_c := \bigl\{\btheta\in\mathbb{R}^K : \norm{\nabla_\btheta\cL_O(\btheta)}^2 \ge c\bigr\}.
\end{equation}
The relevant quantity is the overlap of $p_\btheta$ with $\cS_c$.
 
\begin{lemma}[Overlap lower bound]\label{lem:overlap}
For every prior $r$, generator $f_\bw$, and threshold $c>0$,
\begin{equation*}
    \mathbb{E}_{\bz\sim r}\bigl[\norm{\nabla_\btheta\cL_O(f_\bw(\bz))}^2\bigr]
    = \mathbb{E}_{\btheta\sim p_\btheta}\bigl[\norm{\nabla\cL_O}^2\bigr]
    \ge c\,\cdot\, p_\btheta(\cS_c).
\end{equation*}
\end{lemma}
\begin{proof}
Discard the mass outside $\cS_c$ and bound the integrand below by $c$ on $\cS_c$:
$\mathbb{E}_{p_\btheta}[\norm{\nabla\cL_O}^2] \ge \mathbb{E}_{p_\btheta}[\norm{\nabla\cL_O}^2\,\mathbbm{1}_{\cS_c}] \ge c\,p_\btheta(\cS_c)$.
\end{proof}
 
Barren-plateau avoidance in the spread regime is thus equivalent to placing inverse-polynomial pushforward mass on $\cS_c$ for some $c=1/\mathrm{poly}$, and the prior $r$ is precisely the lever that controls this overlap. For the strategy to be non-vacuous, $\cS_c$ must itself be non-negligible. The next lemma shows that it is: combining the uniform upper bound on the gradient from Lemma~\ref{lem:hessian} with the average lower bound of Ref.~\cite{zhang:2022:gaussian:init} converts the latter from a statement about the mean into one about measure.
 
\begin{lemma}[Coverage of the gradient-rich set under the Gaussian reference]\label{lem:coverage}
Let $\nu=\cN(\boldsymbol 0,\gamma_0^2 I_K)$ be the Gaussian reference of Corollary~\ref{cor:bp-free}, and suppose its guarantee holds in the form $\mu:=\mathbb{E}_{\btheta\sim\nu}\norm{\nabla\cL_O}^2 \ge \mu_0 := 1/\mathrm{poly}(n,L)$. Set $\cS:=\cS_{\mu_0/2}$. Then
\begin{equation*}
    \nu(\cS) \;\ge\; \frac{\mu_0/2}{\,4K-\mu_0/2\,} \;\ge\; \frac{\mu_0}{8K} \;=\; \frac{1}{\mathrm{poly}(n,L)}.
\end{equation*}
\end{lemma}
\begin{proof}
Write $X:=\norm{\nabla\cL_O(\btheta)}^2$ with $\btheta\sim\nu$. By Lemma~\ref{lem:hessian}, $|\partial_j\cL_O|\le 2$, so $0\le X = \sum_{j=1}^K (\partial_j\cL_O)^2 \le 4K =: B$. For any $0\le t<\mu$, splitting the mean over $\{X\le t\}$ and $\{X>t\}$ and bounding $X$ by $t$ and $B$ respectively,
\begin{equation*}
    \mu = \mathbb{E}[X\,\mathbbm{1}_{X\le t}] + \mathbb{E}[X\,\mathbbm{1}_{X>t}] \le t + (B-t)\,\mathbb{P}(X>t),
\end{equation*}
which is the reverse-Markov bound $\mathbb{P}(X>t)\ge (\mu-t)/(B-t)$. Taking $t=\mu_0/2$ and using $\mu\ge\mu_0$ gives $\nu(\cS)\ge\mathbb{P}(X>\mu_0/2)\ge (\mu_0/2)/(B-\mu_0/2)\ge \mu_0/(2B)=\mu_0/(8K)$.
\end{proof}
 
\begin{proposition}[Coverage implies trainability in the spread regime]\label{prop:coverage-bound}
With $\nu$, $\mu_0$, and $\cS$ as in Lemma~\ref{lem:coverage}, suppose a prior $r$ and generator $f_\bw$ produce a pushforward satisfying $p_\btheta(\cS)\ge q$. Then
\begin{equation}
    \mathbb{E}_{\bz\sim r}\bigl[\norm{\nabla_\btheta\cL_O(f_\bw(\bz))}^2\bigr] \;\ge\; \frac{\mu_0}{2}\,q.
\end{equation}
In particular, any LPQC whose pushforward covers $\cS$ with mass $q=1/\mathrm{poly}(n,L)$ is free of barren plateaus, irrespective of how the remaining $1-q$ of the mass is distributed---so the model may remain arbitrarily spread out, and hence expressive, away from $\cS$.
\end{proposition}
\begin{proof}
Apply Lemma~\ref{lem:overlap} with $c=\mu_0/2$ and the hypothesis $p_\btheta(\cS)\ge q$.
\end{proof}
 
Unlike Proposition~\ref{prop:concentrated}, this bound is genuinely prior-dependent: it is driven by where the prior places mass, not by collapsing the pushforward. It remains to exhibit a prior that provably attains $q=1/\mathrm{poly}$. The obstruction is that $\cS$ depends on the unknown target through $\cL_O$; we circumvent it by reserving one mode of the prior as a target-agnostic ``warm start'' whose coverage of $\cS$ is inherited directly from Lemma~\ref{lem:coverage}.
 
\begin{corollary}[Warm-start prior mode]\label{cor:warmstart}
Consider a multimodal prior---equivalently, an MoE generator---whose pushforward is a mixture $p_\btheta = \lambda\,\nu + (1-\lambda)\,\tilde p$ with mixing weight $\lambda$, in which one component is the Gauss-linear branch of the Remark following Corollary~\ref{cor:bp-free} that reproduces $\nu$ exactly (latent dimension $\ge K$ for that branch), and $\tilde p$ is an arbitrary, trainable mixture of the remaining components. Then $p_\btheta(\cS)\ge\lambda\,\nu(\cS)\ge \lambda\mu_0/(8K)$, and by Proposition~\ref{prop:coverage-bound},
\begin{align}
    \mathbb{E}_{\bz\sim r}\bigl[\norm{\nabla_\btheta\cL_O}^2\bigr] \;\ge\; \frac{\lambda\,\mu_0^2}{16K} \;=\; \frac{1}{\mathrm{poly}(n,L)} \nonumber\\ 
    \text{for } \lambda=\Omega\bigl(1/\mathrm{poly}(n,L)\bigr).    
\end{align}

The warm-start mode thus secures trainability target-agnostically, while the remaining components are free to specialize for expressivity.
\end{corollary}
 
This furnishes a theoretical rationale for the architectural choices of the main text: a multimodal prior or a mixture-of-experts generator can dedicate a small fraction of its capacity to maintaining gradient signal while the remaining experts model the target distribution---consistent with the empirical advantage of the MoE configuration in Fig.~\ref{fig:QM9:modes-experts}.
 

\subsection{What remains open}

The results above establish that the LPQC hypothesis class provably contains barren-plateau-free configurations realized by polynomial-size classical generators [(i)--(ii)], and that the meta-landscape is non-flat there [(iii) at initialization]. Two gaps separate this from a complete theory of the empirical behavior in Fig.~\ref{fig:gradnorm:qubits}. First, the guarantee is at (suitably chosen) initialization; extending the lower bound along the entire training trajectory would require controlling how gradient descent on $\bw$ moves the image manifold $\cM_\bw$, which is not currently available. Second, the trained LPQCs in our experiments are not concentrated in a small ball---their pushforwards have $O(1)$ spread (cf.\ the Gauss-tanh curves in Fig.~\ref{fig:gradnorm:qubits})---yet they retain gradient norms one to two orders of magnitude above the broad-measure baselines. Explaining this non-perturbative regime presumably requires connecting the geometry of $\cM_\bw$ to the irreducible-block structure of the DLA~\cite{ragone:2024:lie:bp}, for which the geometric tools of Ref.~\cite{yi:2025:geom:opt:Lie} appear to be the natural starting point. We leave these two extensions to future work, while noting that the empirical evidence in Sec.~\ref{sec:demo} is consistent with the manifold mechanism remaining operative well beyond the perturbative ball regime analyzed here.


\section{Additional Results}\label{appx:results}

\begin{figure*}
  \centering
  \includegraphics[width=2\columnwidth]{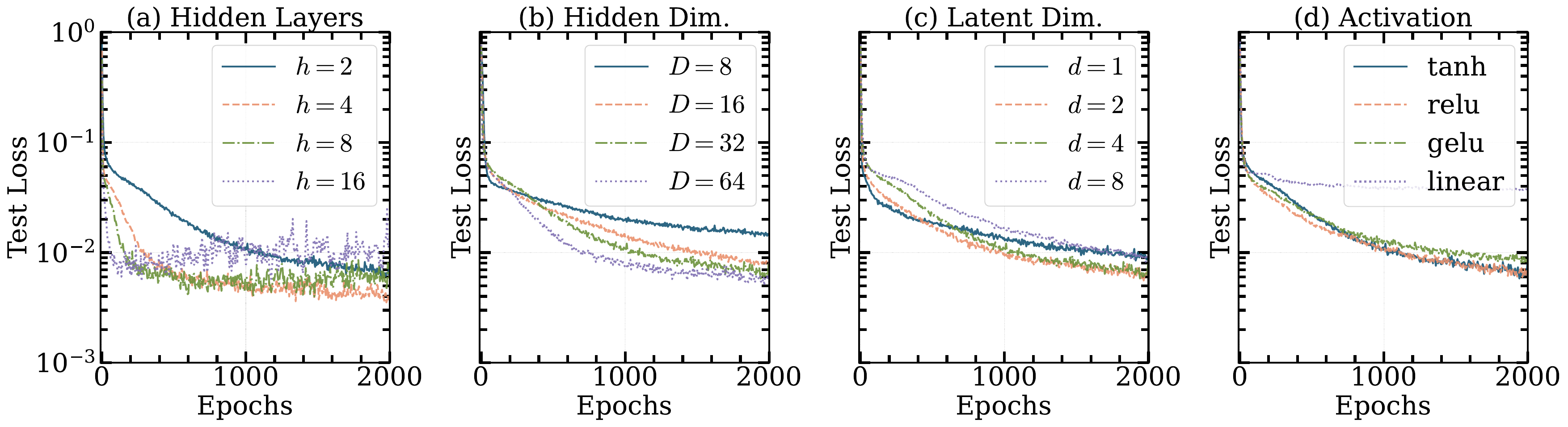}
  \caption{The performance of LPQC models in learning the distribution of multi-clustered density matrices ($n=4$ data qubits) across various configurations of the MLPs $f(\bz; \bw)$ with $h$ hidden layers, $D$ hidden dimension, $d$ latent dimension, and types of the activation function. Results are averaged across 10 experimental trials per configuration. (a) Vary hidden layers $h$ ($D=32$, $d=2$, tanh), (b) Vary hidden dimensions $D$ ($h=2$, $d=4$, tanh),
(c) Vary latent dimensions $d$ ($h=2$, $D=32$, tanh), (d) Vary activation functions ($h=2$, $D=32$, $d=4$). Here, we consider the single-mode Gaussian distribution for $r(\bz)$.}
  \label{fig:compare:MLP}
\end{figure*}

\begin{figure*}
  \centering
  \includegraphics[width=2\columnwidth]{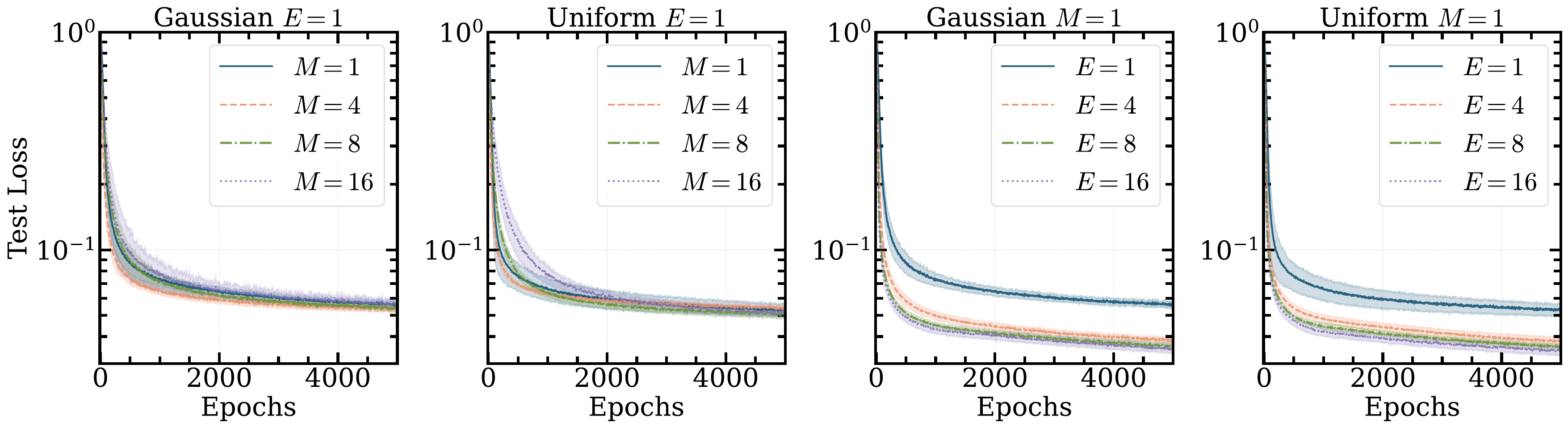}
  \caption{The performance of LPQC in learning the distribution of a QM9-derived distribution (7 heavy atoms and 2 rings in each molecule) across various configurations ($M$ prior modes and $E$ experts)  with $L=10$ PQC layers. Results are averaged across 10 experimental trials per configuration, with solid lines representing the mean and shaded regions indicating the standard deviation.}
  \label{fig:QM9:modes-experts}
\end{figure*}

We present several additional results alongside those in the main text.

Figure~\ref{fig:compare:MLP} illustrates the performance of LPQC models in learning distributions over multi-clustered density matrices, under varying MLP configurations for $f(\bz; \bw)$: number of hidden layers $h$, hidden dimension $D$, latent dimension $d$, and activation functions (tanh, ReLU, and GELU). While activation choices have minimal impact on performance, tuning $h$, $D$, and $d$ plays a pivotal role in enhancing approximation quality. This highlights the importance of an optimal mapping from the prior distribution $r(\bz)$ to the PQC parameter distribution $p(\btheta)$, motivating the use of normalizing flows as a hybrid mechanism to more effectively transform classical latent distributions into quantum ensembles over density matrices.

Figure~\ref{fig:QM9:modes-experts} illustrates the performance of LPQC models in learning the QM9-derived mixed-state distribution (molecules with 7 heavy atoms and 2 rings, as described in the main text) under various configurations of prior modes $M$ and experts $E$, using $L=10$ PQC layers. The results show that increasing the number of fixed modes $M$ in the prior $r(\bz)$ yields only marginal improvements, whereas incorporating a MoE setup enhances performance by approximately 1.5x relative to the single-expert baseline. This disparity arises because the QM9 target distribution lacks a small, well-separated set of modes that align with the fixed, randomly initialized mixture components in the multimodal prior.
By contrast, MoEs dynamically route latent samples $\bz$ to specialized experts (parallel branches in the parameter mapping to $\btheta(\bz)$). This creates an effective, learnable multimodality: even with a unimodal prior, the attention mechanism partitions the latent space into soft regions based on training data, allowing each expert to focus on local structures.

\end{document}